\documentclass[english,10pt]{article}
\usepackage{geometry}
\usepackage{float}
\usepackage{mathrsfs,bm}
\usepackage{amsmath}
\usepackage{amssymb}
\usepackage{fancyhdr}
\usepackage{epsfig}
\usepackage{amsthm}
\usepackage{mathrsfs}
\usepackage{amsfonts}
\usepackage{mathrsfs}
\usepackage{enumerate}
\usepackage{subfigure}

\makeatletter

\floatstyle{ruled}
\newfloat{algorithm}{tbp}{loa}
\providecommand{\algorithmname}{Algorithm}
\floatname{algorithm}{\protect\algorithmname}

\usepackage[auth-sc,affil-it]{authblk}

\newtheorem{prop}{Proposition}[section]

\newcounter{hypA}
\newenvironment{hypA}{\refstepcounter{hypA}\begin{itemize}
  \item[({\bf A\arabic{hypA}})]}{\end{itemize}}

\usepackage{babel}

\makeatother
\begin{document}

\begin{center}

{\Large \textbf{Approximate Bayesian Computation for a Class of Time Series Models}}

\bigskip

BY AJAY JASRA 

{\footnotesize Department of Statistics \& Applied Probability,
National University of Singapore, Singapore, 117546, SG.}\\
{\footnotesize E-Mail:\,}\texttt{\emph{\footnotesize staja@nus.edu.sg}}
\end{center}

\begin{abstract}
In the following article we consider approximate Bayesian computation (ABC) for certain classes of time series models.
In particular, we focus upon scenarios where the likelihoods of the observations and parameter are intractable, by which we mean that 
one cannot evaluate the likelihood even up-to a positive unbiased estimate. This paper reviews and develops a class of approximation
procedures based upon the idea of ABC, but, specifically maintains the probabilistic structure of the original statistical model. This idea is
useful, in that it can facilitate an analysis of the bias of the approximation and the adaptation of established computational methods for
parameter inference. Several existing results in the literature are surveyed and novel developments with regards to computation are given.
\end{abstract}

\section{Introduction} 

Consider an observable, discrete-time stochastic process $\{Y_n\}_{n\geq 1}$, $Y_n\in\mathsf{Y}\subseteq\mathbb{R}^{d_y}$, a latent and unobserved discrete-time stochastic process $\{X_n\}_{n\geq 1}$,
$X_n\in\mathsf{X}\subseteq\mathbb{R}^{d_x}$, a parameter $\theta\in\Theta\subseteq\mathbb{R}^{d_{\theta}}$ and i.i.d.~sequence of random variables $\{\phi_n\}_{n\geq 1}$, $\phi_n\in\Phi\subseteq\mathbb{R}^{d_{\phi}}$, whose distribution can depend upon $\theta$ ($d_y,d_x,d_{\theta},d_{\phi}\in\mathbb{N}$). This article is concerned with the class
of statistical models, for $ n\geq 1$, $\varphi_{n,\theta}:\mathbb{N}\times\Theta\times\mathsf{Y}^{n-1}\times\mathsf{X}\rightarrow\mathsf{Y}$
$$
Y_n = \varphi_{n,\theta}(y_{1:n-1},x_n,\phi_n)
$$
where $y_{1:n-1}:=(y_1,\dots,y_{n-1})$, $n\geq 2$, $y_{1:-1}$ is assumed null, 
which induces a joint Lebesgue (assumed for simplicty of presentation) density of the model
\begin{equation}
p_{\theta}(y_{1:n},x_{1:n}) = \Big(\prod_{i=1}^n p_{\theta}(y_i|y_{1:i-1},x_i)\Big)p_{\theta}(x_{1:n})\label{eq:model}
\end{equation}
where we use $p_{\theta}(\cdot)$ to denote conditional and joint probability densities w.r.t.~Lebesgue measure and $P_{\theta}(\cdot)$ will denote the associated distribution. This collection of models is rather flexible
and contains:
\begin{enumerate}
\item{I.I.D.~models, when $x_{1:n}$ is null and $\varphi_{n,\theta}$ does not depend on $y_{1:n-1}$}
\item{Observation-driven time series models (ODTS) (e.g.~\cite{cox}),  when $x_{1:n}$ is null}
\item{Hidden Markov models (HMMs) (e.g.~\cite{cappe}), when $x_{1:n}$ follows a Markov chain and $\varphi_{n,\theta}$ does not depend on $y_{1:n-1}$}
\item{Some Non-Linear time series models (e.g.~\cite{fan}), such as causal and invertible bilinear time series models \cite{granger,liu} (when $x_{1:n}$ is null).}
\end{enumerate}
The class of models is quite large and includes many popular classes including some GARCH models (\cite{bollerslev}), stochastic volatility models (e.g.~\cite{kim}) and partially observed Markov jump 
process models (e.g.~\cite{golightly}). The list of applications is numerous ranging from economics, biology and engineering; the reader is referred to the afore mentioned references for specific applications and details.

The particular scenario of interest in this article is when one can exactly  sample the process $\{\phi_n\}_{n\geq 1}$ and evaluate for $ n\geq 1$, $\varphi_{n,\theta}$,
but that the density function $p_{\theta}(y_i|y_{1:i-1},x_i)$ is unknown up-to a positive and unbiased estimate (this can occur see \cite{jacob}); we will call the likelihood intractable. 
If  the density function $p_{\theta}(y_i|y_{1:i-1},x_i)$ is known up-to a positive and unbiased estimate
and suppose that this estimate is
$\prod_{i=1}^n\hat{p}_{\theta}(y_i|y_{1:i-1},x_i,Z_i)$ for some random variables $\{Z_n\}_{n\geq 1}$ where $\mathbb{E}[\prod_{i=1}^n\hat{p}_{\theta}(y_i|y_{1:i-1},x_i,Z_i)]=
\prod_{i=1}^n p_{\theta}(y_i|y_{1:i-1},x_i)$ (the expectation is w.r.t.~the law of $\{Z_n\}_{n\geq 1}$), then the approximation schemes we will discuss are not always needed as 
exact (possibly Monte Carlo-based) inference procedures are possible. Later in the article, we will consider cases when $p_{\theta}(x_{1:n})$ is intractable in some way, but for now,
we will assume, when it is in the model, that this density is known point-wise up-to a constant. We also remark, in Section \ref{sec:approx}, we will show that if one can evaluate
the density of the noise terms $p_{\theta}(\phi_i)$ then, indeed one does not require the ability to sample the $\{\phi_n\}$.
 For the class of problems of interest, for either likelihood-based or Bayesian inference
methods, the complexity of the model is such that even when using advanced methods such Markov chain Monte Carlo (MCMC) or sequential Monte Carlo (SMC) exact statistical
inference is seldom possible. One of the standard solutions to this problem is to introduce an approximation of the statistical model and an often adopted approach is that of using 
Approximate Bayesian Computation, especially when taking a Bayesian perspective to statistical estimation; see for instance \cite{marin} for a recent overview.

Before continuing, it is noted that the scenario of intractable likelihoods will occur in many applications. Perhaps the most common is when $\{\phi_n\}_{n\geq 1}$ is distributed according to a model
which can be simulated, but $p_{\theta}(\phi)$ is not known; a good example is the stable distribution - this would find applications for i.i.d.~models and HMMs (see \cite{peters,yildrlim}) applied in finanical
contexts where the heavy tails of these distributions are a realistic modelling choice. Other models used in practice include the Lotka-Voltera model \cite{boys} used for stochastic-kinetic networks and the $g-and-k$ distribution used for `non-standard' data \cite{haynes}.

\subsection{Approximate Bayesian Computation}

Suppose that one places a prior density $\pi$ on $\theta\in\Theta$, the standard ABC approximation (see e.g.~\cite{pritchard}) of the posterior associated to the likelihood in \eqref{eq:model}:
$$
\pi(\theta,x_{1:n}|y_{1:n}) \propto p_{\theta}(y_{1:n},x_{1:n})\pi(\theta)
$$
is to take, for $\epsilon>0$:
\begin{equation}
\pi^{\epsilon}(\theta,x_{1:n},u_{1:n}|y_{1:n}) \propto \mathbb{I}_{\{u_{1:n}\in\mathsf{Y}^n:\mathsf{d}(s(u_{1:n}),s(y_{1:n}))<\epsilon\}}(u_{1:n})p_{\theta}(u_{1:n},x_{1:n})\pi(\theta) \label{eq:abc_post}
\end{equation}
where $u_{1:n}\in\mathsf{Y}^n$ are \emph{auxiliary} data, $p_{\theta}(u_{1:n},x_{1:n})$ is as \eqref{eq:model}, $S:\mathsf{Y}^n\rightarrow\mathbb{R}^{d_s}$ is a \emph{summary} of the data, $\mathsf{d}:\mathbb{R}^{2d_s}\rightarrow\mathbb{R}_+$ is \emph{distance} on the summary statistics. If the summary statistics are sufficient for the data, then one can show (under minimal assumptions) that as $\epsilon \downarrow 0$ that the (marginal) ABC posterior $\pi^{\epsilon}(\theta,x_{1:n}|y_{1:n})$ is in some sense exactly $\pi(\theta,x_{1:n}|y_{1:n})$. 
It is remarked that if $\mathsf{Y}$ is countably infinite (we assume this is not the case), then it is possible to take $\epsilon=0$ and if $S$ is sufficient, this leads to \emph{exact} inference as proposed in \cite{rubin}.
We note that the indicator function is used in \eqref{eq:abc_post}, but it can be replaced with a kernel density; we focus upon indicator functions in this article. The approximation \eqref{eq:abc_post} is particularly amenable to inference, in that, in the framework of this article
one can sample $u_{1:n}|x_{1:n},\theta$ exactly from the true model; this fact is very useful in that it would allow one to apply a number of computational algorithms such as rejection sampling, importance sampling, MCMC or SMC (this explained in detail later in the article). The problems with this approximation include:
\begin{itemize}
\item{If the statistics $s$ are not sufficient, one does not always recover the exact posterior, even if $\epsilon=0$. This can make it difficult to characterize, mathematically, the bias in approximation. Selecting summary statistics can then be challenging, although, there are approaches such as \cite{fearnhead}.}
\item{In many practical models, there is a trade-off between making $\epsilon$ small (accuracy of the approximation) and allowing a computational method to work well (accuracy in computation). The approximation \eqref{eq:abc_post} does not often help this former trade-off due to the factor $\mathbb{I}_{\{u_{1:n}\in\mathsf{Y}^n:\mathsf{d}(s(u_{1:n}),s(y_{1:n}))<\epsilon\}}(u_{1:n})$, which may mean that $\epsilon$ has to be relatively large for a computational algorithm to produce reliable results. For example in a rejection algorithm which samples $\theta,x_{1:n},u_{1:n}$ from $\pi(\theta)p_{\theta}(u_{1:n},x_{1:n})$  the acceptance probability is $\mathbb{I}_{\{u_{1:n}\in\mathsf{Y}^n:\mathsf{d}(s(u_{1:n}),s(y_{1:n}))<\epsilon\}}(u_{1:n})$; if $n$ is large one may need to make $\epsilon$ quite large to yield reasonable acceptance rates.}
\item{In the context of this article, \eqref{eq:abc_post} does not retain the probabilistic structure of the model. For example one does not have
$$
\pi^{\epsilon}(\theta,x_{1:n}|y_{1:n}) \propto \Big(\prod_{i=1}^n 
p_{\theta}^{\epsilon}(y_i|y_{1:i-1},x_i)
\Big)
p_{\theta}(x_{1:n}) \pi(\theta) 
$$
for some probability density $p_{\theta}^{\epsilon}(y_i|y_{1:i-1},x_i)$. The issue with this is that the techniques of inference for the original model (e.g.~for HMMs one often uses SMC methods) cannot be used, without modification, which creates an additional difficulty. Another issue is that a theoretical study of the ABC approximation may be far more difficult than is the case when the probabilistic structure of the model is retained. Although there are several results in the literature (e.g.~\cite{blum}) one could argue that a more precise analysis could be undertaken by retaining the structure of the model.}
\end{itemize}
The points here have been made in many other articles including \cite{chopin1,bart,ehrlich,jasra1,jasra}. Thus our objective is to use a different ABC approximation which, to an extent, can deal with some of the deficiencies raised here. We note that there are a great deal of extensions of the standard ABC approach (such as regression adjustment \cite{beaumont}), but the argument made in this article, is that the approximation to be reviewed and developed here, is very reasonable, especially for the case when $d_y$ is low to moderate.

The ABC-based approximation we consider is as follows, which can deal with some of the issues outlined above
$$
\pi^{\epsilon}(\theta,x_{1:n}|y_{1:n}) \propto p_{\theta}^{\epsilon}(y_{1:n},x_{1:n})\pi(\theta)
$$
where 
\begin{equation}
p_{\theta}^{\epsilon}(y_{1:n},x_{1:n}) = \Big(\prod_{i=1}^n 
p_{\theta}^{\epsilon}(y_i|y_{1:i-1},x_i)
\Big)
p_{\theta}(x_{1:n})\label{eq:abc}
\end{equation}
and
$$
p_{\theta}^{\epsilon}(y_i|y_{1:i-1},x_i)
= \frac{\int_{\mathsf{Y}} \mathbb{I}_{\{u_i\in\mathsf{Y}:\mathsf{d}(u_i,y_i)<\epsilon\}}(u_i)p_{\theta}(u_i|y_{1:i-1},x_i)du_i}{\int_{\{u\in\mathsf{Y}:\mathsf{d}(u,y_i)<\epsilon\}}du}.
$$
Typically, one needs to be able to evaluate the ABC posterior pointwise up-to a normalizing constant (for example for Monte Carlo algorithms) which is often not the case for 
$\pi^{\epsilon}(\theta,x_{1:n}|y_{1:n})$. \emph{One particular choice} is to define an ABC posterior on an extended state-space from $\Theta\times\mathsf{X}^n$ to $\Theta\times\mathsf{X}^n\times\mathsf{Y}^n$, setting
\begin{eqnarray*}
\pi^{\epsilon}(\theta,x_{1:n},u_{1:n}|y_{1:n}) & \propto & p_{\theta}^{\epsilon}(y_{1:n},u_{1:n},x_{1:n})\pi(\theta)\\
p_{\theta}^{\epsilon}(y_{1:n},u_{1:n},x_{1:n}) & = & \Big(\prod_{i=1}^n 
\frac{\mathbb{I}_{\{u_i\in\mathsf{Y}:\mathsf{d}(u_i,y_i)<\epsilon\}}(u_i)p_{\theta}(u_i|y_{1:i-1},x_i)}{\int_{\{u\in\mathsf{Y}:\mathsf{d}(u,y_i)<\epsilon\}}du}
\Big)
p_{\theta}(x_{1:n}).
\end{eqnarray*}
The probabilistic structure in \eqref{eq:model} is retained in \eqref{eq:abc}. We note that the issue in the second bullet point for problems of standard ABC approximations is dealt with to an extent;
it breaks the global dependence (on the data i.e.~the factor $\mathbb{I}_{\{u_{1:n}\in\mathsf{Y}^n:\mathsf{d}(s(u_{1:n}),s(y_{1:n}))<\epsilon\}}(u_{1:n})$) into a collection of smaller and possibly easier to handle (from a computational perspective) sub-problems
(i.e.~the factors $\mathbb{I}_{\{u_i\in\mathsf{Y}:\mathsf{d}(u_i,y_i)<\epsilon\}}(u_i)$).
This type of approximation has been studied in various special cases in \cite{chopin1,calvet,dean,dean1,delmoraljacod,ehrlich,
jasra1,jasra2,jasra,martin,mckcoodea2009,theo}, the actual intepretation will allow one, as we will explain later on in this article, to perform asymptotically in $n$ (for some classes of models and under some assumptions and modifications) consistent parameter estimation when taking a likelihood based (i.e.~trying to optimize $\log[p_{\theta}^{\epsilon}(y_{1:n},u_{1:n},x_{1:n})]$) or Bayesian approach to this task. The approximation does not use any summary statistics; this is only an option if $d_y$ is small to moderate - this latter scenario is sufficiently rich in that 
one does not always attempt to fit the original model (i.e.~approximation is not needed) if this were not the case.

This article will consider the ABC approximation \eqref{eq:abc} for a variety of special cases (that is, the models 1.-3.~mentioned above). We will review and discuss both the inferential
aspects associated to this approximation (such as consistency in parameter estimation, bias and so fourth) and the computational methods that can be used to fit the models. Several
original remarks are made along the way. It is remarked that a comprehensive comparison of (some) computational methods for ABC approximations posteriors based upon \eqref{eq:abc}
and those based upon \eqref{eq:abc_post} can be found in \cite{chopin1}.

\subsection{Structure of the Article}

The article is structured as follows. 
In Section \ref{sec:rem_abc} we investigate our ABC approximation in more details, discuss the notion of ABC and some alternative approximation schemes.
In Section \ref{sec:models} we consider three classes of models which include i.i.d.~models, ODTS models and HMMs. We
discuss the ABC approximation and adjustments which can allow one to perform asymptotically consistent parameter estimation. In Section \ref{sec:comp} we give a variety of computational
methods which can be used to fit the particular ABC approximations discussed in Section \ref{sec:models}. We mainly consider SMC and MCMC methods, of which we assume the reader
has some familiarity. Other approaches such as expectation-propagation are also considered. In Section \ref{sec:summary} the article is concluded and
future research work is discussed. The appendix houses some proofs and assumptions of some propositions which are given in the article.

\section{Some Remarks on ABC}\label{sec:rem_abc}

\subsection{The Approximation}\label{sec:approx}

We begin our review by making some remarks on the approximation in \eqref{eq:abc}. 
We first show, using arguments from \cite{biau} that the approximation bias will disappear as $\epsilon\rightarrow 0$.
Throughout $\mathsf{d}$ is the Euclidean distance; it does not appear that in practice the selection of the distance criterion makes a significant difference.
We have the following result for any bounded, measurable real-valued function $\xi:\Theta\times\mathsf{X}^n\rightarrow\mathbb{R}$. The assumption (A\ref{hyp:1}) along with the proof is in Appendix \ref{sec:prfs}.

\begin{prop}\label{prop:ep_conv}
Assume (A\ref{hyp:1}). Then for almost every $y_{1:n}\in\mathsf{Y}^n$
$$
\lim_{\epsilon\rightarrow 0}\int_{\Theta\times\mathsf{X}^n}\xi(\theta,x_{1:n}) \pi^{\epsilon}(\theta,x_{1:n}|y_{1:n}) dx_{1:n}d\theta = \int_{\Theta\times\mathsf{X}^n}\xi(\theta,x_{1:n}) \pi(\theta,x_{1:n}|y_{1:n})  dx_{1:n} d\theta.
$$
\end{prop}
The result is fairly standard in the literature and is included for completeness of the article as well as for pedagogical purposes. Essentially it tells us, at the level of approximation,
but before any fitting is considered, that one would like to make $\epsilon$ as small as possible.

The extended ABC posterior, to be defined below, is needed for computational algorithms, as one does not know $\pi^{\epsilon}(\theta,x_{1:n}|y_{1:n})$ point-wise up-to a normalizing constant;
this is the main requirement of most Monte Carlo based simulation methods, which are often needed to draw inference from the ABC posterior. An important remark concerning the extended ABC posterior
$$
\pi^{\epsilon}(\theta,x_{1:n},u_{1:n}|y_{1:n}) \propto p_{\theta}^{\epsilon}(y_{1:n},u_{1:n},x_{1:n})\pi(\theta)
$$
where $p_{\theta}^{\epsilon}(y_{1:n},u_{1:n},x_{1:n})$ is as equation \eqref{eq:abc} 
is not the only choice which yields $\pi^{\epsilon}(\theta,x_{1:n}|y_{1:n})$ as a marginal where
$$
\pi^{\epsilon}(\theta,x_{1:n}|y_{1:n}) = \int_{\mathsf{Y^n}} \pi^{\epsilon}(\theta,x_{1:n},u_{1:n}|y_{1:n}) du_{1:n}.
$$ 
One possible alternative is the density 
$$
p_{\theta}^{\epsilon}(y_{1:n},\phi_{1:n},x_{1:n}) = \Big(\prod_{i=1}^n 
\frac{\mathbb{I}_{\{\phi_i\in\Phi:\mathsf{d}(\varphi_{i,\theta}(y_{1:i-1},x_i,\phi_i),y_i)<\epsilon\}}(\phi_i)p_{\theta}(\phi_i)}{\int_{\{\phi\in\Phi:\mathsf{d}(\varphi_{i,\theta}(y_{1:i-1},x_i,\phi),y_i)<\epsilon\}}d\phi}
\Big)
p_{\theta}(x_{1:n})
$$
which we term the \emph{collapsed} representation of the model, paralleling this idea for HMMs (see \cite{murray}). The subsequent ABC posterior is
$
\pi^{\epsilon}(\theta,x_{1:n},\phi_{1:n}|y_{1:n}) \propto p_{\theta}^{\epsilon}(y_{1:n},\phi_{1:n},x_{1:n}) \pi(\theta).
$
This type of decomposition, as discussed in both \cite{andrieu1,yildrlim} in the context of ABC is particularly interesting,
indeed, if $p_{\theta}(\phi)$ is known point-wise up-to a constant, then (as remarked by \cite{andrieu1}) \emph{one no longer requires the ability to sample the data}, that is, to sample the $\{\phi_n\}_{n\geq 1}$. This is because, assuming one can evaluate the prior $\pi(\theta)$,
the ABC posterior can be evaluated point-wise and up-to a normalizing constant.
It is remarked that the collapsed representation might be bourne out of \emph{inferential}  considerations, such as inability to sample the data.
The collapsed representation will be particularly useful for HMMs as we will discuss in Section \ref{sec:hmms}. 

In Section \ref{sec:comp} we will see different extended ABC posteriors that can yield computational algorithms which are more efficient (in some sense)
than if one considered $\pi^{\epsilon}(\theta,x_{1:n},u_{1:n}|y_{1:n})$. It should be noted that this is a \emph{computational} consideration and not an inferential one; that is computational
improvements can be yielded and not necessarily inferential ones. 
At this stage, we do not dwell on these points and give extended ABC posteriors $\pi^{\epsilon}(\theta,x_{1:n},u_{1:n}|y_{1:n})$ (or the collapsed version)
in our subsequent discussions (up-to Section \ref{sec:comp}). The reader, should, however, keep these ideas in mind.

\subsection{Noisy ABC}\label{sec:noisy_abc}

An interesting idea, developed in \cite{fearnhead,wil2008} is the notion of \emph{noisy} ABC, a term apparently coined in \cite{fearnhead}, see also \cite{dean}. The idea, which follows \cite[Theorem 1]{wil2008}, is this; the model which is 
to be fitted is:
$$
\pi(\theta,x_{1:n}|y_{1:n}) \propto \Big(\prod_{i=1}^n 
p_{\theta}^{\epsilon}(y_i|y_{1:i-1},x_i)
\Big)
p_{\theta}(x_{1:n}) \pi(\theta).
$$
The idea is then if the data \emph{are from the original model} i.e.~that in \eqref{eq:model} (although this is not a pre-requisite of using the approach), then instead of using the original observations $y_{1:n}$, one uses perturbed observations defined as, for $i\geq 1$
$$
Z_i | y_{1:i} \sim \mathcal{U}_{\{z\in\mathsf{Y}:\mathsf{d}(z,y_i)<\epsilon\}}
$$
where $\mathcal{U}_A$ is the uniform distribution on a set $A$. Then, one fits the model
$$
\pi(\theta,x_{1:n}|z_{1:n}) \propto \Big(\prod_{i=1}^n 
p_{\theta}^{\epsilon}(z_i|z_{1:i-1},x_i)
\Big)
p_{\theta}(x_{1:n}) \pi(\theta).
$$
using the techniques that were to be adopted when using the original data. The intuition is that the perturbed observations are now from the model that is actually being fitted and one hopes for favourable statistical properties of estimates of $\theta$
and the associated posterior (as well as the $x_{1:n}$ if they are of interest).

From a mathematical perspective \cite[Theorem 2]{fearnhead} show that noisy ABC is calibrated, that is, roughly, that the noisy ABC posterior has good properties when the number of data are large (although as pointed out by \cite{robert}, there appear to be some missing
technicalities in their result). Similar results are derived in \cite{dean1,dean} for HMMs except when concerning both Bayesian and classical estimation, that is computing 
\begin{equation}
\widetilde{\theta}_{n}^{\epsilon} = \textrm{argmax}_{\theta\in\Theta}\Bigg(\frac{1}{n}\log\Big(\int_{\mathsf{X}^n} \Big(\prod_{i=1}^n 
p_{\theta}^{\epsilon}(z_i|z_{1:i-1},x_i)
\Big)
p_{\theta}(x_{1:n}) dx_{1:n}\Big)\Bigg) \label{eq:noisy_mle}.
\end{equation}
In particular consistency and asymptotic normality properties are investigated. The main point is that (for HMMs and under assumptions) the noisy ABC MLE in \eqref{eq:noisy_mle} is consistent (that is, returns the true parameter that generated the data, $\theta^{\star}$ say) as $n\rightarrow\infty$, whatever $\epsilon>0$ is chosen. We will investigate this idea in more details in the subsequent
sections of the paper.

\subsection{Alternatives to ABC}

At its core, ABC is simply a way to approximate the posterior density and one may question why to choose ABC versus many other alternatives, such as indirect inference \cite{gour}.
The purpose of this article is not to go into comparisons between various approximation schemes (that has been done in many articles such as \cite{fearnhead,jasra}), but we make
some mention of alternatives here. Perhaps the main competitor to ABC is indirect inference; \cite[Theorem 5]{fearnhead} indicates that in some scenarios ABC can be more accurate as
$\epsilon\rightarrow 0$. 

In the context of HMMs there are a variety of alternatives such as nonparametric particle filtering \cite{camros2009,gauchi}; a comparison with ABC is discussed
in \cite{jasra}. For HMMs there is also a link between ABC and the ensemble Kalman filter (e.g.~\cite{evensen}) see \cite{nott}, which is an approximation of the filtering density of an HMM (this is discussed
later on): a direct comparison of these ideas does not appear to be in the literature. An approach proposed by for cases where the observation model is `unknown' is developed in \cite{kanagawa}, but these scenario are 
even more complex than encountered in this article.

 Some ideas which are based upon ABC, but different to those considered in this article can be found in \cite{adam,theo}. 
In \cite{adam} the authors form a kernel density estimate of the
ABC likelihood from $\theta$ samples drawn from the ABC posterior distribution of $\theta$. They then
maximise this kernel density estimate as an approximation to MLE.
In \cite{theo}, they use an ABC approximation based upon \eqref{eq:abc} (except with a non-negative kernel),
and use a similar idea to \cite{adam} in building a kernel density estimate except for the posterior on $\theta$.

One technique found in the data assimilation literature, termed randomized maximum likelihood (e.g.~\cite{oliver}), involves repeated simulations from the prior and likelihood, to obtain estimators of parameters and is, to an extent,
related to ABC (although is not necessarily adopted when the likelihood is intractable). To our knowledge, this approach has not been systematically understood, nor compared to ABC.

\section{Models}\label{sec:models}

\subsection{I.I.D.~Models}

\subsubsection{Approximation}\label{sec:iid_approx}

In the context that there is no latent process $\{X_n\}$ and $\varphi_{n,\theta}$ does not depend upon $Y_{1:n-1}$ one has the following version of the ABC approximation in 
\eqref{eq:abc}
$
p_{\theta}^{\epsilon}(y_{1:n}) = \Big(\prod_{i=1}^n 
p_{\theta}^{\epsilon}(y_i)
\Big)
$
where
\begin{equation}
p_{\theta}^{\epsilon}(y_i)
= \frac{\int_{\mathsf{Y}} \mathbb{I}_{\{u_i\in\mathsf{Y}:\mathsf{d}(u_i,y_i)<\epsilon\}}(u_i)p_{\theta}(u_i)du_i}{\int_{\{u\in\mathsf{Y}:\mathsf{d}(u,y_i)<\epsilon\}}du}.\label{eq:iid_abc_like}
\end{equation}
yielding the extended ABC posterior:
\begin{equation}
\pi^{\epsilon}(\theta,u_{1:n}|y_{1:n}) \propto \Big(\prod_{i=1}^n\mathbb{I}_{\{u_i\in\mathsf{Y}:\mathsf{d}(u_i,y_i)<\epsilon\}}(u_i)p_{\theta}(u_i)\Big)\pi(\theta)\label{eq:abc_iid_ex}.
\end{equation}
One can also use collapsed versions as in Section \ref{sec:approx} and replace the $y_{1:n}$ with $z_{1:n}$ as in Section \ref{sec:noisy_abc}.
Approximations of these type  have been used in \cite{chopin1,yildrlim} and include $\alpha-$stable models \cite{peters} where the observations possess an $\alpha-$stable
distribution and the $g-and-k$  distribution. In the former model \cite{peters} show that it is difficult in the ABC approximation \eqref{eq:abc_post} to select summary statistics;
the utility of the ABC approximation based upon \eqref{eq:abc} is thus evident, where this problem as been removed.
In general this ABC approximation is expected to be of interest if one can sample the $\{\phi_n\}_{n\geq 1}$ but cannot evaluate the associated
density, or if one cannot sample the $\{\phi_n\}_{n\geq 1}$ but can evaluate the associated density (for which the collapsed representation could be used). We now further discuss why this approximation might be useful in practice.

\subsubsection{Properties of the Approximation}

One reason why the ABC approximation based around \eqref{eq:abc}  for this and other models is quite reasonable is as follows. Suppose that one is only interested in likelihood-based inference on
$\theta$ (we discuss Bayesian inference below), then one would seek to find the MLE:
$$
\widehat{\theta}_n^{\epsilon} = \textrm{argmax}_{\theta}\Big(\frac{1}{n}\sum_{i=1}^n \log(p_{\theta}^{\epsilon}(y_i))\Big).
$$
As mentioned above, one could also consider the noisy ABC MLE:
$$
\widetilde{\theta}_n^{\epsilon} = \textrm{argmax}_{\theta\in\Theta}\Big(\frac{1}{n}\sum_{i=1}^n \log(p_{\theta}^{\epsilon}(z_i))\Big).
$$
Now suppose that the data $Y_i$ are distributed according to the true model i.e.~$Y_i\stackrel{\textrm{i.i.d.}}{\sim}P_{\theta^{\star}}(\cdot)$, $\theta^{\star}\in\Theta$.
Also define
$$
\theta^{\star,\epsilon} = \textrm{argmax}_{\theta\in\Theta}\int_{\mathsf{Y}}\log(p_{\theta}^{\epsilon}(y)) p_{\theta^*}(y) dy.
$$
Write $\stackrel{a.s.}{\rightarrow}$ to denote almost sure convergence as $n\rightarrow\infty$.
Then we have the following result for i.i.d.~models that mirrors results in \cite{dean1,dean,jasra1} and whose proof and assumptions can be found in Appendix \ref{sec:prfs_abc}.

\begin{prop}\label{prop:abc_ok}
Assume (A\ref{hyp:2}). Then for any $\epsilon>0$
\begin{itemize}
\item{$\widehat{\theta}_n^{\epsilon}\stackrel{a.s.}{\rightarrow}\theta^{\star,\epsilon}$}
\item{$\widetilde{\theta}_n^{\epsilon}\stackrel{a.s.}{\rightarrow}\theta^{\star}$.}
\end{itemize}
\end{prop}

The result shows that the MLE associated to standard ABC will converge to a point $\theta^{\star,\epsilon}$ which \emph{may} be different than the true parameter; to show that it is always different requires some work, see for instance \cite[Theorem 1]{dean} for HMMs.
Converesely the noisy ABC MLE will return the \emph{true} parameter as the number of data grow. There are a number of issues to be discussed about results of this type. Firstly, the result is assuming that the data are being generated from the model under study, which some
investigators may not be prepared to assume; in such scenarios one would then prefer to use ABC against noisy ABC as one relies on the true data instead of perturbing it (see also \cite{fearnhead} for some discussion). Secondly, the result shows that whilst the noisy ABC estimate will, for a large number of data, give the 
true parameter, it does not characterize whether there is any loss in efficiency in using noisy ABC versus the true MLE (which one cannot obtain as the likelihood is intractable). This can be achieved by considering asymptotic normality and the Fisher information matrix; we direct the reader
towards \cite{dean} where in the context of HMMs it is shown that there is a loss of efficiency and its dependence upon $\epsilon$ is characterized. 
Thirdly, the result does not take into account that one can seldom compute the estimators analytically; so the importance of these points may be reduced. Finally, the result does not say anything about the ABC posterior. The result is easily extended to
the MAP estimator, if the prior density on $\theta$ is bounded away from zero (which could happen, as we have assumed $\Theta$ is compact). More generally one may want to consider Bayesian consistency and Bernstein-Von Mises theorems; we direct the reader to \cite{dean1}
for more discussion on this issue. 

The main point to take on board from Proposition \ref{prop:abc_ok} is that the ABC approximation is quite reasonable, in that very sensible parameter estimates can be yielded. Inspection of the proofs help to support some of the points made in the introduction;
the removal of summary statistics combined with retaining the probabilistic structure of the likelihood has allowed one to discover some of the underlying properties of the ABC approximation. One could extend Proposition \ref{prop:abc_ok} to include summary
statistics as done in \cite{dean}.

To illustrate the point in Proposition \ref{prop:abc_ok}, we provide a numerical example that considers the model: $Y_n=\theta+\phi_n$, where $\phi_n\stackrel{\textrm{i.i.d.}}{\sim}\mathcal{N}(0,1)$ where $\mathcal{N}(0,1)$
is the standard normal distribution in one-dimension. We assume that \emph{a priori} $\theta\sim\mathcal{N}(0,1)$. We generated two data-sets of size 100 and 1000 from the model (when $\theta=0$) and whilst ABC is not needed here (the posterior is known analytically)
we fit a standard and noisy ABC approximation to the data. We run Algorithm \ref{alg:new} (in Section \ref{sec:comp}) and plot a density estimate of the samples in Figure \ref{fig:fig1}. The plot shows, reasonably clearly that the MAP estimator
for noisy ABC is closer to that of the true parameter (which is zero) as $n$ grows, than that of the standard ABC MAP.

\begin{figure}[h]
\centering \subfigure[$n=100,\epsilon=10$]{{\includegraphics[width=0.49\textwidth,height=7.5cm]{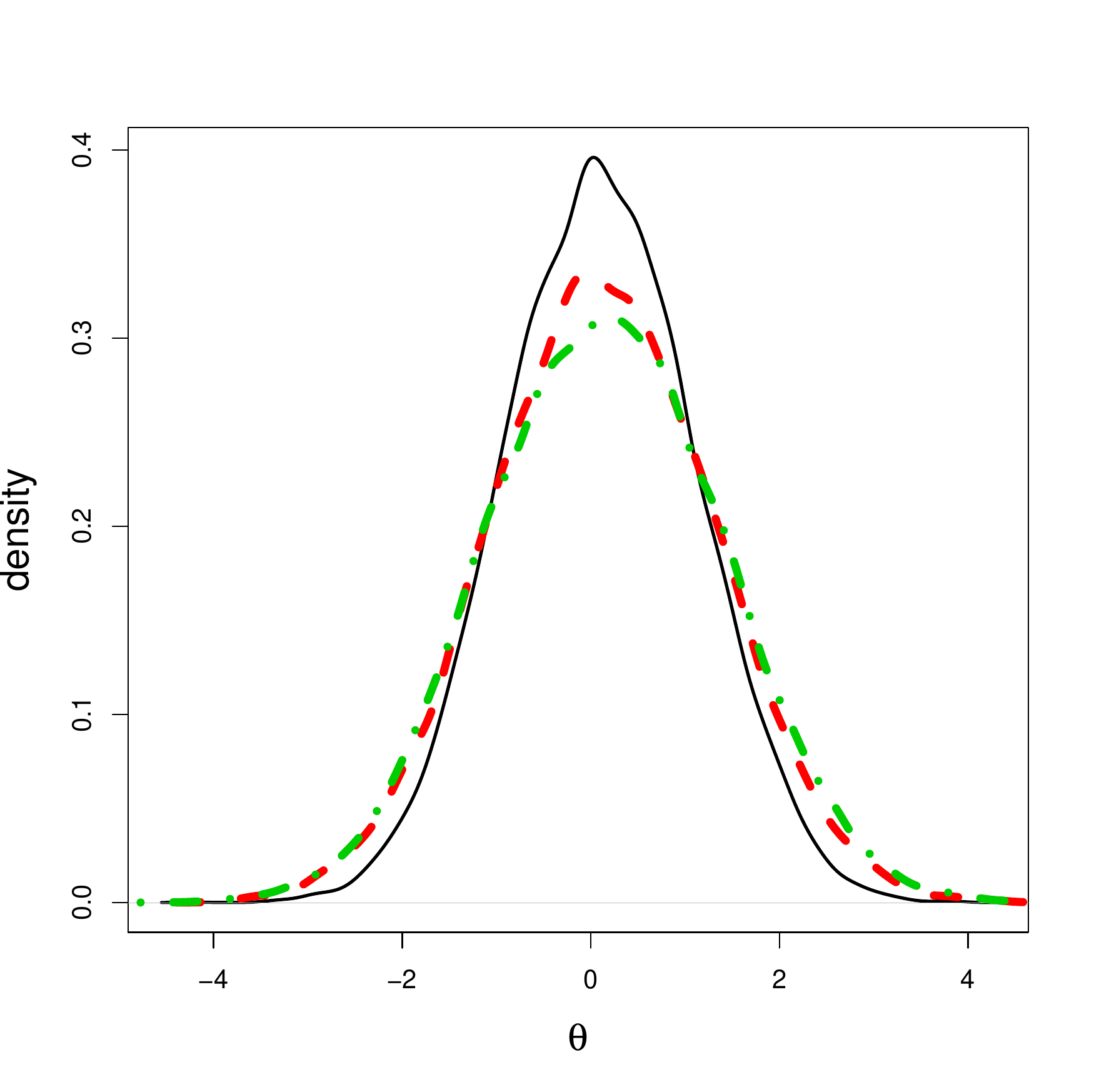}}}
\subfigure[$n=1000,\epsilon=10$]{{\includegraphics[width=0.49\textwidth,height=7.5cm]{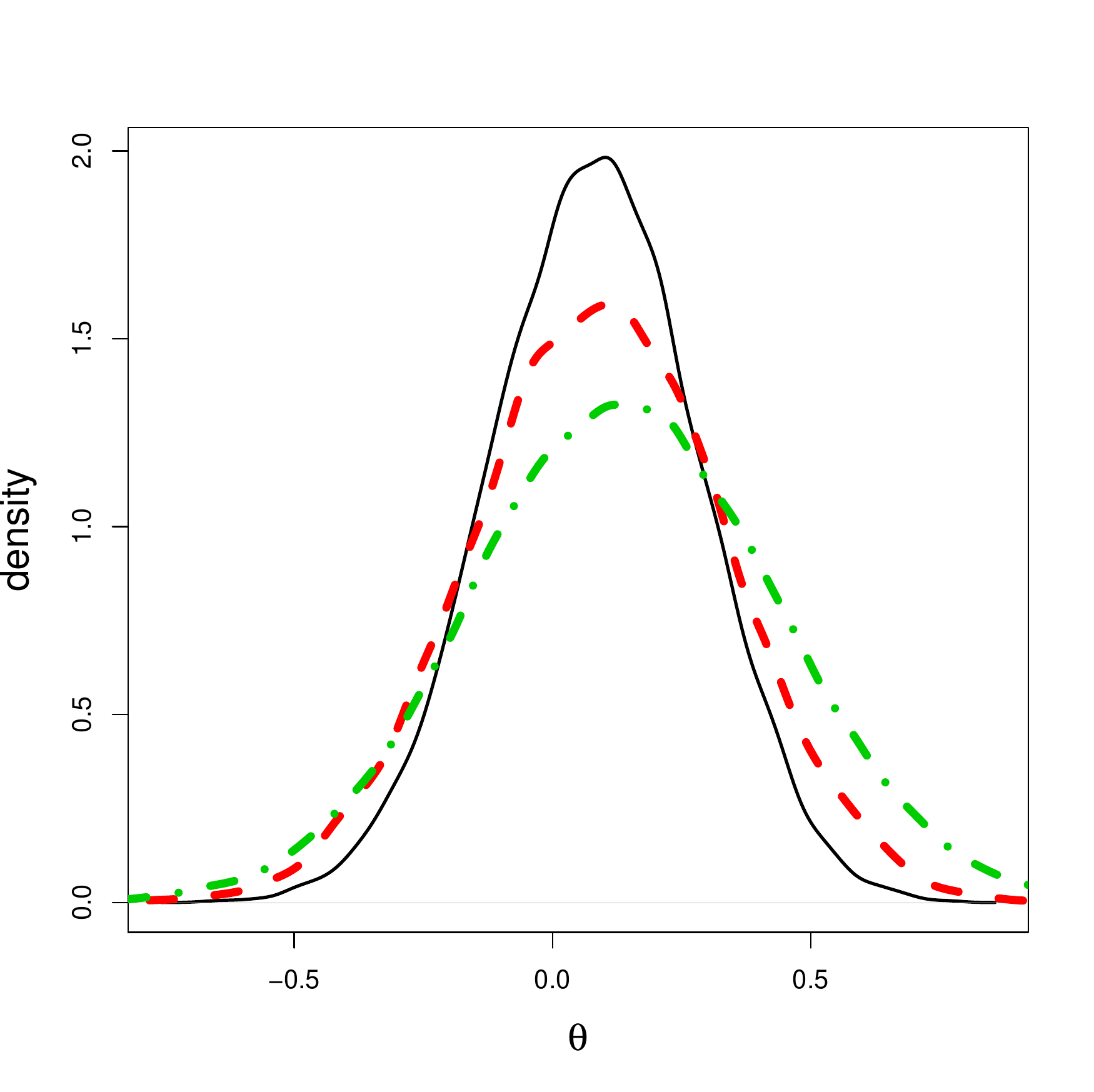}}}
\caption{Marginal MCMC Density Plots for Normal Means Example. The thin density line is the true posterior, the noisy ABC posterior is the red dash line
and the standard ABC posterior is the dash-dot green line.}
\label{fig:fig1}
\end{figure}

\subsection{Observation-Driven Time Series}

\subsubsection{Model and Approximation}

The model can be described as follows. We observe $\{Y_n\}_{n\in\mathbb{N}}$, which are associated to an unobserved process $\{X_n\}_{n\in \mathbb{N}}$ which is potentially unknown.
Define the process $\{Y_k,X_k\}_{k\in\mathbb{N}_0}$ (with $y_0$ some arbitrary point on $\mathsf{Y}$) and $X_0$ to be introduced below. Denote by $\mathscr{F}_k$ as the information in $\{Y_n,X_n\}$ up-to time $k$.
The model is defined as, for $n\in\mathbb{N}_0$ (using $\mathbb{P}$ to denote probability) 
\begin{eqnarray*}
\mathbb{P}_{\gamma}(Y_{n+1}\in A|\mathscr{F}_n) & = & \int_A p_{\gamma}(y_{n+1}|x_n) dy_{n+1} \quad A\subseteq\mathsf{Y}\\
X_{n+1} & = & \Psi^{\gamma}(X_n,Y_{n+1}) \\
\mathbb{P}_{\gamma}(X_0\in B) & = & \int_{B} \nu_{\gamma}(x_0)dx_0  \quad B\subseteq\mathsf{Y}
\end{eqnarray*}
where $\gamma\in\Gamma\subseteq\mathbb{R}^{d_{\gamma}}$,
$\Psi:\Gamma\times\mathsf{X}\times\mathsf{Y}\rightarrow\mathsf{X}$ and $\nu_{\gamma}(x_0)$ is a probability
density on $\mathsf{X}$ for every $\gamma\in\Gamma$. 
Next, we define a prior probability density $\vartheta(\gamma)$ and write 
$\pi(x_0,\gamma)=\nu_{\gamma}(x_0)\vartheta(\gamma)$ assumed to be a proper
joint probability density on $\mathsf{X}\times\Gamma$. Thus, given $n$ observations the object of inference is the posterior distribution on $\Gamma\times \mathsf{X}$:
\begin{equation}
\pi(\gamma,x_0|y_{1:n}) \propto \Bigg(\prod_{i=1}^n p_{\gamma}(y_i|\Psi^{\gamma}_{i-1}(y_{0:i-1},x_0))\Bigg)\pi(x_0,\gamma) 
\label{eq:true_post}
\end{equation}
where we have used the notation for $i>1$, $\Psi_{i-1}^{\gamma}(y_{0:i-1},x_0) = \Psi^\gamma \circ\cdots\circ \Psi^{\gamma}(x_0,y_1)$,
$\Psi_1^{\gamma}(y_0,x_0):=\Psi^{\gamma}(x_0,y_0)$,
$\Psi_{0}^{\gamma}(x_0,y_{0}):=x_0$. Note that $p_{\gamma}(y_{i}|x_{i-1})$ is a Lebesgue density assumed to be induced by the relationship $Y_i=\varphi_{\gamma}(\Psi^{\gamma}_{i-1}(y_{0:i-1},x_0),\phi_n)$
for $\{\phi_n\}_{n\geq 1}$ and i.i.d.~sequence, and whilst $\{X_n\}_{n\geq 0}$ is a latent process one can consider $\theta=(\gamma,x_0)$, $\Theta=\Gamma\times\mathsf{X}$, so that the model falls into the framework of this article.
  In most applications of practical interest, the posterior cannot be computed point-wise and one has to resort
to numerical methods, such as MCMC, to draw inference on $\theta$ and/or $x_0$. The models which lie under this class are some GARCH models and some Bayesian inverse problems
(see e.g.~\cite{kantas} and the references therein).

The extended ABC approximation for this model is
$$
\pi^{\epsilon}(x_{0},\gamma,u_{1:n}|y_{1:n})\propto \Big(\prod_{i=1}^{n}\mathbb{I}_{\{u:\mathsf{Y}:\mathsf{d}(u,y_i)<\epsilon\}}(u_{i})p_{\gamma}(u_i|\Psi_{i-1}^{\gamma}(y_{0:i-1},x_{0}))\Big)\pi(x_{0},\gamma)
$$
which is not dissimilar to the i.i.d.~case. Indeed, one can use similar computational methods to the i.i.d.~case to fit the model. It is also remarked that one can introduce a collapsed representation as well as noisy ABC for this class of models.

\subsubsection{Properties of the Approximation}

The ABC approximation for ODTS models has a similar structure to the i.i.d.~model. However, to prove results of the type in Proposition \ref{prop:abc_ok} one has to appeal to the consistency theory for such models which can be found in 
\cite{douc}. This is done in \cite{jasra1}, but under very strong assumptions which imply that $\mathsf{Y},\Theta$ are compact state-spaces. Indeed, there are numerous results which could be proved (but do not appear to be) for ABC approximations of such classes of models, including
Bayesian consistency and Bernstein von-Mises theorems.

\subsection{Hidden Markov Models}\label{sec:hmms}

\subsubsection{Model and Approximation}

Hidden Markov Models are widely used in statistics; see \cite{cappe} for a recent overview. An HMM\ is a pair of discrete-time stochastic
processes, $\left\{  X_{n}\right\}  _{n\geq1}$ and $\left\{
Y_{n}\right\} _{n\geq 1}$. The hidden process
$\left\{  X_{n}\right\}_{n\geq 1} $ is a\ Markov chain with $x_0$ a given point and transition density $p_{\theta}(x_{n}|x_{n-1})$, i.e
\begin{equation}
\mathbb{P}(X_{n}\in A|X_{n-1}=x_{n-1})=\int_{A}p_{\theta}(x_{n}|x_{n-1})dx_{n}\quad n\geq1 \label{eq:evol}%
\end{equation}
where $A\subseteq\mathsf{X}$.
The observation
$Y_{n}$ is assumed to be conditionally independent of other variables given $\left\{  X_{n}\right\}_{n\geq 0}$ and its conditional distribution is specified by a density
$p_{\theta}(y_n|x_{n})$
\begin{equation}
\mathbb{P}(Y_{n}\in B|\{X_{k}\}_{k\geq 0}=\{x_{k}\}_{k\geq 0})=\int_{B}p_{\theta}(y_n|x_{n})dy_{n}\quad n\geq 1  \label{eq:obs}%
\end{equation}
with $B\subseteq\mathsf{Y}$. When the parameters of the model are fixed, particular interest is in the filtering density
$$
\pi_{\theta}(x_n|y_{1:n}) \propto \int_{\mathsf{X}^{n-1}} \prod_{i=1}^{n} p_{\theta}(y_i|x_i)p_{\theta}(x_i|x_{i-1}) dx_{1:n-1}
$$
and joint smoothing density
$$
\pi_{\theta}(x_{1:n}|y_{1:n}) \propto  \prod_{i=1}^{n} p_{\theta}(y_i|x_i)p_{\theta}(x_i|x_{i-1}).
$$
We note that in general, expectations w.r.t.~the associated distributions are not known analytically (with the exception of a few special cases), even if
$ p_{\theta}(y_i|x_i)p_{\theta}(x_i|x_{i-1})$ can be evaluated. Typically SMC methods are used to approximate these latter distributions.

The extended ABC approximation, which, to our knowledge first appeared in \cite{delmoraljacod} (when $\theta$ is fixed), is
\begin{equation}
\pi^{\epsilon}(\theta,x_{1:n},u_{1:n}|y_{1:n}) \propto \Big(\prod_{i=1}^{n} \mathbb{I}_{\{u\in\mathsf{Y}:\mathsf{d}(u,y_i)<\epsilon\}}(u_i)p_{\theta}(u_i|x_i)p_{\theta}(x_i|x_{i-1})\Big) \pi(\theta)\label{eq:abc_hmm_post}
\end{equation}
it should be noted that the densities in \eqref{eq:evol} and \eqref{eq:obs} can be time-inhomogeneous, but we have adopted a time-homogeneous representation for simplicity of notation.
The representation has been used in many articles; a non-exhaustive list is \cite{calvet,dean1,dean,ehrlich,jasra2,jasra,martin,mckcoodea2009}. A collapsed representation may be adopted
and is discussed further in Section \ref{sec:collapse} and, as before one can easily adopt a noisy ABC approach. A variety of models fall into the class of HMMs, including stochastic volatility models.

\subsubsection{Properties of the Approximation}

The ABC approximation of HMMs appears to be the most extensively studied, of all the approximations considered in this article. Some theoretical results are as follows. When one considers
$\theta$ fixed and uses the original data (i.e.~standard ABC) \cite{jasra} show under strong assumptions that for a bounded measurable $\xi:\mathsf{X}\rightarrow\mathbb{R}$
$$
\Big|\int_{\mathsf{X}}\xi(x) [\pi_{\theta}(x_n|y_{1:n})-\pi_{\theta}^{\epsilon}(x_n|y_{1:n})] dx_n\Big| \leq C\epsilon\sup_{x\in\mathsf{X}}|\xi(x)|
$$
where $C$ does not depend on $n$. That is to say, that the bias of the ABC approximation of the filter does not necessarily increase with $n$ and can be controlled by making $\epsilon$
arbitrarily small. Similarly \cite{martin} show under assumptions that for $\xi_1,\dots,\xi_n$ bounded, measurable with $\xi_i:\mathsf{X}\rightarrow\mathbb{R}$
$$
\Big|\int_{\mathsf{X}^n}\Big(\sum_{i=1}^n\xi_i(x_i)\Big) [\pi_{\theta}(x_{1:n}|y_{1:n})-\pi_{\theta}^{\epsilon}(x_{1:n}|y_{1:n})] dx_{1:n}\Big| \leq Cn\epsilon\max_i\Big(\sup_{x\in\mathsf{X}}|\xi_i(x)|)\Big)
$$
where $C$ does not depend on $n$, 
which says that the bias of the ABC approximation of the expectation of additive functions w.r.t.~the smoothing distribution grows at most linearly with respect to $n$ and could be controlled by making $\epsilon$ suitably small. Results of this type are of particular interest in smoothing; see e.g.~\cite{dds1} for further details.
Additional results associated to the ABC bias of the filter derivative, that is $|\nabla_{\theta} [\int_{\mathsf{X}}\xi(x) [\pi_{\theta}(x_n|y_{1:n})-\pi_{\theta}^{\epsilon}(x_n|y_{1:n})] dx_n]|$ can be found in \cite{dean1,dean,ehrlich}; these latter results can be used in the proofs of the results to be discussed below, or of
independent interest of themselves (see e.g.~\cite{poya} why this is the case).

More interesting, perhaps are results regarding parameter estimation. In the context of Proposition \ref{prop:abc_ok}, \cite{dean} have shown (under assumptions) that the MLE associated to standard ABC (i.e.~$\widehat{\theta}_n^{\epsilon}=\textrm{argmax}_{\theta\in\Theta}\frac{1}{n}\log[p_{\theta}^{\epsilon}(y_{1:n})]$) has an intrinsic bias as $n$ grows, indeed they find the bias as a function of $\epsilon$.
This result is more precise than that provided in Proposition \ref{prop:abc_ok} which does not show that the ABC MLE is necessarily different from the true parameter and does not characterize the asymptotic bias. \cite{dean} then show that the noisy ABC MLE
 (i.e.~$\widetilde{\theta}_n^{\epsilon}=\textrm{argmax}_{\theta\in\Theta}\frac{1}{n}\log[p_{\theta}^{\epsilon}(z_{1:n})]$)  is asymptotically consistent and
then give an asymptotic normality result
$$
\sqrt{n}(\widetilde{\theta}_n^{\epsilon}-\theta^*) \Rightarrow \mathcal{N}_{d_{\theta}}(0,I(\epsilon))
$$
where $\Rightarrow$ denotes convergence in distribution as $n$ grows and $\mathcal{N}_{d_{\theta}}(0,I(\epsilon))$ denotes the $d_{\theta}-$dimensional normal distribution with zero mean vector and covariance matrix $I(\epsilon)$. 
The authors go even further, in that they compare $I(\epsilon)$ with the Fisher information matrix associated to the MLE; the loss of efficiency (the price to pay for using an ABC approximation) is given as a function of $\epsilon$.
These results are further refined in \cite{dean1} and Bayesian consistency and Bernstein von-Mises theorems are considered for the the noisy ABC posterior.

\subsubsection{The Collapsed Representation}\label{sec:collapse}

The collapsed representation (as discussed in Section \ref{sec:approx}) is particularly interesting for HMMs. A version can be found in \cite{yildrlim}, but we will present something slightly different, which allows one to deal with scenarios
where the transition density of the hidden Markov chain $p_{\theta}(x_n|x_{n-1})$ is also intractable. We will suppose that the hidden Markov chain satisfies
$$
X_n = \rho_{\theta}(x_{n-1},\eta_n) \quad n\geq 1
$$
where $\{\eta_n\}$ is an i.i.d.~sequence of random variables $\eta_n\in\Xi\subseteq\mathbb{R}^{d_{\eta}}$ and $\rho:\Theta\times\mathsf{X}\times\Xi\rightarrow\mathsf{X}$.
Not all hidden Markov chains will fall in this class, but it is sufficiently large to be of practical interest. Then, one may adopt the following ABC approximation
$$
\pi^{\epsilon}(\theta,\phi_{1:n},\eta_{1:n}|y_{1:n}) \propto \Big(\prod_{i=1}^{n} \mathbb{I}_{\{\phi\in\Phi:\mathsf{d}(\varphi_{i,\theta}(x_n,\phi),y_i)<\epsilon\}}(\phi_i)p_{\theta}(\phi_i)p_{\theta}(\eta_i)\Big) \pi(\theta)
$$
where we have that $x_n=\rho_{\theta}(\rho_{\theta}(\circ\cdots\circ\rho_{\theta}(x_0,\eta_1),\eta_{n-1}),\eta_n)$. The key point here is that one can then estimate so-called doubly intractable HMMs if either $p_{\theta}(\phi_i)p_{\theta}(\eta_i)$
are known or one can sample them both (or a combination of these). This idea can be found in \cite{zhang} and relies upon the approach in \cite{murray}.

\section{Computation}\label{sec:comp}

\subsection{Introduction}

We will now describe a collection of computational methods which can be used to fit the ABC approximations considered in the previous section. We note that the computation of ODTS models
is much the same as for i.i.d.~models and thus we only consider i.i.d.~models. The reader is reminded that some familiarity with SMC and MCMC methods is assumed; complete introductions can be found in \cite{doucet,robert1}. We also remark that at this stage, at the level of ABC approximations, one would like to make $\epsilon$ as close as possible
to zero; we shall see (and as discussed for example in \cite{fearnhead}) that there is a trade-off in making $\epsilon$ small and allowing the computational method to provide reliable results.
One point is that the computational methods are designed for the scenario where an indicator function is used in the ABC approximation. Some mention of what can be done when non-negative kernels
are used instead will be given. Note also, that all the algorithms are given when using the observed data $y_{1:n}$; the approach for noisy ABC is exactly the same except one replaces $y_{1:n}$ with $z_{1:n}$.

\subsection{MCMC Methods for I.I.D.~Models}\label{sec:mcmc_for_iid}

We will now consider sampling from $\pi^{\epsilon}(\theta|y_{1:n})\propto p_{\theta}^{\epsilon}(y_{1:n})\pi(\theta)$, $p_{\theta}^{\epsilon}(y_{1:n})=\prod_{i=1}^np_{\theta}^{\epsilon}(y_i)$, where
$p_{\theta}^{\epsilon}(y_i)$ is given in \eqref{eq:iid_abc_like}. An MCMC algorithm, is given  in Algorithm \ref{alg:mcmc}. In general this approach is not feasible because one does not know the term
$p_{\theta}^{\epsilon}(y_{1:n})$. The idea is to advocate algorithms by constructing targets on an auxiliary state-space (such as \eqref{eq:abc_iid_ex}) which can efficiently approximate this marginal MCMC approach;
this is the idea of the psuedo marginal \cite{andrieu2}. The point is to construct MCMC algorithms which can replicate the properties of the (hopefully) efficient marginal MCMC algorithm.
Note that the following ideas can be easily extended to the collapsed representation discussed in Section \ref{sec:approx}. Also note that some of this discussion can be found in \cite{jasra1}, where even more details are given.

\begin{algorithm}
\begin{enumerate}
\item \textbf{(Initialisation)} At $t=0$ sample $\theta_{0}$ from the prior.
\item \textbf{(M-H kernel)} For $t\geq1$:

\begin{itemize}
\item Sample $\theta'|\theta_{t-1}$ from a proposal $Q(\cdot|\theta_{t-1})$ with density
$q(\theta'|\theta_{t-1})$.
\item Accept the proposed state and set $\theta_{t}=\theta'$ with probability
\[
1\wedge\frac{p_{\theta'}^{\epsilon}(y_{1:n})}{p_{\theta_{t-1}}^{\epsilon}(y_{1:n})}\times\frac{\pi(\theta')q(\theta_{t-1}|\theta')}{\pi(\theta_{t-1})q(\theta'|\theta_{t-1})},
\]
otherwise set $\theta_{t}=\theta_{t-1}$. Set $t=t+1$ and return to the start of 2.
\end{itemize}
\end{enumerate}
\caption{A marginal M-H algorithm for $\pi^{\epsilon}(\theta|y_{1:n})$}
\label{alg:mcmc}
\end{algorithm}

\subsubsection{A Naive MCMC Algorithm}

Initially consider the ABC approximation when extended to the space
$\Theta\times\mathsf{Y}^{n}$: 
\[
\pi^{\epsilon}(\theta,u_{1:n}|y_{1:n}) \propto \Big(\prod_{i=1}^n\mathbb{I}_{\{u_i\in\mathsf{Y}:\mathsf{d}(u_i,y_i)<\epsilon\}}(u_i)p_{\theta}(u_i)\Big)\pi(\theta).
\]
 In Algorithm \ref{alg:simple}
we present a natural Metropolis-Hastings (M-H) proposal (which is essentially the MCMC method of \cite{majoram}) that could be used to sample from
$\pi^{\epsilon}(\theta,u_{1:n}|y_{1:n})$. The one-step transition kernel of the MCMC chain is usually described
as the \emph{M-H kernel} and follows from Step 2 in Algorithm \ref{alg:mcmc}. It should be noted that one can implement `early rejection' (e.g.~\cite{lee1}), in that if for some $i$
$u_i'\notin\{u\in\mathsf{Y}:\mathsf{d}(u,y_i)<\epsilon\}$ then the proposal can be terminated, as it will be rejected.

What is hopefully apparent from Algorithm \ref{alg:simple} is that, interpreting the Hastings ratio as an importance weight, the variance of this weight will play a critical
role in the efficiency of the algorithm. That is, in Algorithm \ref{alg:simple} the Hastings ratio is
$$
\frac{\prod_{i=1}^{n}\mathbb{I}_{ \{u\in\mathsf{Y}:\mathsf{d}(u,y_i)<\epsilon\} }(u_{i}^{'})}{\prod_{i=1}^{n}\mathbb{I}_{ \{u\in\mathsf{Y}:\mathsf{d}(u,y_i)<\epsilon\} }(u_i)}\times\frac{\pi(\theta')q(\theta|\theta')}{\pi(\theta)q(\theta'|\theta)}
$$
which is of a similar form to the Hastings ratio in Algorithm \ref{alg:mcmc}, except it replaces $p_{\theta'}^{\epsilon}(y_{1:n})/p_{\theta_{t-1}}^{\epsilon}(y_{1:n})$ with
estimates 
\begin{equation}
\frac{\prod_{i=1}^{n}\mathbb{I}_{ \{u\in\mathsf{Y}:\mathsf{d}(u,y_i)<\epsilon\} }(u_{i}^{'})}{\prod_{i=1}^{n}\mathbb{I}_{ \{u\in\mathsf{Y}:\mathsf{d}(u,y_i)<\epsilon\} }(u_i)}\label{eq:pseduo_ratio}.
\end{equation}
If one supposes that Algorithm \ref{alg:mcmc} is quite efficient, then one would hope to replicate similar properties, which will occur if the variability of \eqref{eq:pseduo_ratio} is not too large, or that the individual estimates (in the ratio) are quite well behaved.
These individual estimates in the ratio are (up-to a constant) unbiased estimates (expectation w.r.t.~$\prod_{i=1}^{n}p_{\theta}(u_{i})$) of $p_{\theta}^{\epsilon}(y_{1:n})$ and it is this property which allows Algorithm \ref{alg:simple} to provide
$\theta$ samples from the marginal distribution of interest.
As we have remarked previously
the choice $\pi^{\epsilon}(\theta,u_{1:n}|y_{1:n})$ is one amongst potentially many densities which will give our ABC posterior as a marginal.
That is, if one can find a different extended target, so that the variance (say for fixed $\theta$ - although that is another aspect of importance)
of the estimate of $p_{\theta}^{\epsilon}(y_{1:n})$ is smaller, then one might expect that the associated MCMC algorithm will mix better.

As noted, for example in \cite{jasra1},
as $n$ increases, the M-H kernel in Algorithm \ref{alg:simple}
will have an acceptance probability that falls quickly with $n$.
In particular, for any fixed $\theta,\epsilon$, the probability of obtaining
a sample with a non-zero acceptance probability
will fall at an exponential rate in $n$. This means
that this basic ABC MCMC approach will be inefficient for moderate
values of $n$.

\begin{algorithm}
\begin{itemize}
\item {Sample $\theta'|\theta$ from a proposal $Q(\cdot|\theta)$ with
density $q(\theta'|\theta)$.} 
\item {Sample $u_{1:n}^{'}$ from a distribution with joint density $\prod_{i=1}^{n}p_{\theta'}(u_{i})$
} 
\item {Accept the proposed state $\left(\theta',u'_{1:n}\right)$ with
probability: 
\[
1\wedge\frac{\prod_{i=1}^{n}\mathbb{I}_{ \{u\in\mathsf{Y}:\mathsf{d}(u,y_i)<\epsilon\} }(u_{i}^{'})}{\prod_{i=1}^{n}\mathbb{I}_{ \{u\in\mathsf{Y}:\mathsf{d}(u,y_i)<\epsilon\} }(u_i)}\times\frac{\pi(\theta')q(\theta|\theta')}{\pi(\theta)q(\theta'|\theta)}.
\]
 } 
\end{itemize}
\caption{M-H Proposal for basic ABC MCMC}
\label{alg:simple}
\end{algorithm}

\subsubsection{MCMC with $N$ Trials}

The above issue can be dealt with by using $N$ multiple trials (see for instance \cite{delmoral} in a SMC context), so that
for each data point the chance of getting an auxiliary data point in $\{u\in\mathsf{Y}:\mathsf{d}(u,y_i)<\epsilon\}$ is larger. This approach augments the posterior to a larger state-space, $\Theta\times\mathsf{Y}^{nN}$,
in order to target   the following density: 
\[
\widetilde{\pi}^{\epsilon}(\theta,u_{1:n}^{1:N}|y_{1:n})= \frac{\pi(\theta)p_{\theta}^{\epsilon}(y_{1:n})}{\int\pi(\theta)p_{\theta}^{\epsilon}(y_{1:n})d\theta}\frac{1}{p_{\theta}^{\epsilon}(y_{1:n})} \Big(\prod_{i=1}^{n}\frac{\sum_{j=1}^{N}\mathbb{I}_{ \{u\in\mathsf{Y}:\mathsf{d}(u,y_i)<\epsilon\} } (u_i^j)}
{N\int_{\{u\in\mathsf{Y}:\mathsf{d}(u,y_i)<\epsilon\}} du}\Big)
\prod_{i=1}^{n}\prod_{j=1}^{N}p_{\theta}(u_{i}^{j}).
\]
One can show that the marginal of interest $\pi^{\epsilon}(\theta|y_{1:n})$
is preserved, i.e. 
\[
\pi^{\epsilon}(\theta|y_{1:n})=\int_{\mathsf{Y}^{nN}}\widetilde{\pi}^{\epsilon}(\theta,u_{1:n}^{1:N}|y_{1:n})du_{1:n}^{1:N}=\int_{\mathsf{Y}^{n}}\pi^{\epsilon}(\theta,u_{1:n}|y_{1:n})du_{1:n}.
\]
In Algorithm \ref{alg:Ntry} we present an M-H kernel with invariant
density $\widetilde{\pi}^{\epsilon}$.  It is expected, especially with $N$ large,
that this algorithm should be more efficient than in Algorithm \ref{alg:simple} (which is the case $N=1$).
This is because as $N$ grows one expects 
$$
\Big(\prod_{i=1}^{n}\frac{\sum_{j=1}^{N}\mathbb{I}_{ \{u\in\mathsf{Y}:\mathsf{d}(u,y_i)<\epsilon\} } (u_i^j)}
{N\int_{\{u\in\mathsf{Y}:\mathsf{d}(u,y_i)<\epsilon\}} du}\Big)
$$
to be close to $p_{\theta}^{\epsilon}(y_{1:n})$ and hence the marginal algorithm described at the beginning of Section \ref{sec:mcmc_for_iid}.
It is noted in \cite{jasra1} that to ensure that the (relative) variance (w.r.t.~$\prod_{i=1}^{n}\prod_{j=1}^{N}p_{\theta}(u_{i}^{j})$) of the estimate
of $p_{\theta}^{\epsilon}(y_{1:n})$ does not depend upon $n$ (possibly the major contributor to the variance is $n$ growing),
one should set $N=\mathcal{O}(n)$.

It has been shown in \cite{lee1}
that even the M-H kernel in Algorithm \ref{alg:Ntry} does not always
perform well. It can happen that the chain gets often stuck in regions
of the state-space $\Theta$ where 
\[
\alpha_{i}(y_i,\epsilon,\theta):=\int_{\{u\in\mathsf{Y}:\mathsf{d}(u,y_i)<\epsilon\}}p_{\theta}(u)du
\]
is small. Note that the $\alpha_{i}(y_i,\epsilon,\theta)$ could be larger by making $\epsilon$ larger, but this is at the cost that the ABC approximation is worse; before resorting to this, one might try to improve upon the algorithm.
Given this notation, we remark that
$$
p_{\theta}^{\epsilon}(y_{1:n}) = \prod_{i=1}^n\frac{\alpha_{i}(y_{i},\epsilon,\theta)}{\int_{\{u\in\mathsf{Y}:\mathsf{d}(u,y_i)<\epsilon\}}du}.
$$

\begin{algorithm}
\begin{itemize}
\item {Sample $\theta'|\theta$ from a proposal $Q(\cdot|\theta)$ with
density $q(\theta'|\theta)$.} 
\item {Sample $\left.u'\right._{1:n}^{1:N}$ from a distribution with joint
density $\prod_{i=1}^{n}\prod_{j=1}^{N}p_{\theta'}(u_{i}^{j})$
.} 
\item {Accept the proposed state $(\theta',\left.u'\right._{1:n}^{1:N})$
with probability: 
\[
1\wedge\frac{\prod_{i=1}^{n}(\frac{1}{N}\sum_{j=1}^{N}\mathbb{I}_{\{u\in\mathsf{Y}:\mathsf{d}(u,y_i)<\epsilon\}}(\left.u'\right._{i}^{j}))}{\prod_{i=1}^{n}(\frac{1}{N}\sum_{j=1}^{N}\mathbb{I}_{\{u\in\mathsf{Y}:\mathsf{d}(u,y_i)<\epsilon\}}(u_{i}^{j}))}\times\frac{\pi(\theta')q(\theta|\theta')}{\pi(\theta)q(\theta'|\theta)}.
\]
 } 
\end{itemize}
\caption{M-H Proposal for ABC with N trials}
\label{alg:Ntry}
\end{algorithm}

\subsubsection{MCMC with a Random Number of Trials}

The following MCMC kernel is developed in \cite{jasra1} and is based on the $N-$hit kernel
of \cite{lee}, which was adapted to account for the data structure in the model.
One of the problems of the MCMC kernel in Algorithm \ref{alg:Ntry} is that when $\alpha_{i}(y_i,\epsilon,\theta)$ is small,
the algorithm gets stuck. The approach is that, when this occurs one tries to use more computational effort, without necessarily being more clever.
The idea is that instead of sampling $N$ samples for each data point $i$, the procedure samples data-points until there are $N$ points which
lie in $\{u\in\mathsf{Y}:\mathsf{d}(u,y_i)<\epsilon\}$. \cite{jasra1} show how this is possible by defining an appropriate target density, for $N\geq2$, $m_{k}\in\{N,N+1,\dots,\}$,
$1\leq k\leq n$: 
$$
\hat{\pi}^{\epsilon}(\theta,m_{1:n}|y_{1:n})=  
$$
$$
\frac{\pi(\theta)p_{\theta}^{\epsilon}(y_{1:n})}{\int\pi(\theta)p_{\theta}^{\epsilon}(y_{1:n})d\theta}\frac{1}{p_{\theta}^{\epsilon}(y_{1:n})}
\prod_{i=1}^{n}\frac{N-1}{(m_{i}-1)\int_{\{u\in\mathsf{Y}:\mathsf{d}(u,y_i)<\epsilon\}}du}
\prod_{i=1}^{n}\binom{m_{i}-1}{N-1}\alpha_{i}(y_{i},\theta,\epsilon)^{N}(1-\alpha_{i}(y_{i},\theta,\epsilon))^{m_{i}-N}.
$$
It is shown in \cite{jasra1} that
that the marginal w.r.t. $\theta$ is the one of interest: 
\[
\pi^{\epsilon}(\theta|y_{1:n})=\sum_{m_{1:n}\in\{N,N+1,\dots\}^{n}}\widehat{\pi}_n^{\epsilon}(\theta,m_{1:n}|y_{1:n}).
\]

The MCMC kernel is given in Algorithm \ref{alg:new}. \cite{jasra1} show that one should set $N=\mathcal{O}(n)$ to control the relative variance (w.r.t.~$n$) w.r.t.~$\prod_{i=1}^{n}\binom{m_{i}-1}{N-1}\alpha_{i}(y_{i},\theta,\epsilon)^{N}(1-\alpha_{i}(y_{i},\theta,\epsilon))^{m_{i}-N}$ of the estimate of $p_{\theta}^{\epsilon}(y_{1:n})$ 
which is 
$$
\prod_{i=1}^{n}\frac{N-1}{(m_{i}-1)\int_{\{u\in\mathsf{Y}:\mathsf{d}(u,y_i)<\epsilon\}}du}.
$$

\begin{algorithm}
\begin{itemize}
\item {Sample $\theta'|\theta$ from a proposal $Q(\cdot|\theta)$ with
density $q(\theta'|\theta)$.} 
\item {For $i=1,\dots,n$ repeat the following: sample $u_{i}^{1},u_{i}^{2},\dots$
with probability density $p_{\theta'}(u_{i})$
until there are $N$ samples lying in $\{u\in\mathsf{Y}:\mathsf{d}(u,y_i)<\epsilon\}$; the
number of samples to achieve this (including the successful trial)
is $m_{k}'$. } 
\item {Accept $\left(\theta',m_{1:n}'\right)$ with probability: 
\[
1\wedge\frac{\prod_{i=1}^{n}\frac{1}{m_{i}'-1}}{\prod_{i=1}^{n}\frac{1}{m_{i}-1}}\times\frac{\pi(\theta')q(\theta|\theta')}{\pi(\theta)q(\theta'|\theta)}.
\]
 } 
\end{itemize}
\caption{M-H Proposal with a random number of trials}
\label{alg:new}
\end{algorithm}

\subsubsection{Simulations}\label{sec:simos_iid}

We will consider the following ODTS model, also discussed in \cite{jasra1}.
Set, for $(Y_{n},X_{n})\in\mathbb{R}\times\mathbb{R}^{+}$ 
\begin{eqnarray*}
Y_{n+1} & = & \kappa_{n}\quad n\in\mathbb{N}_{0}\\
X_{n+1} & = & \beta_{0}+\beta_{1}X_{n}+\beta_{2}Y_{n+1}^{2}\quad n\in\mathbb{N}_{0}
\end{eqnarray*}
 where $\kappa_{n}|x_{n}\stackrel{\textrm{ind}}{\sim}\mathcal{S}(0,x_{n}, s_{1},s_{2})$
(i.e.~a stable distribution, with location 0, scale $X_{n}$ and
asymmetry and skewness parameters $s_{1},s_{2}$; see \cite{chambers} for more information). We
set 
\[
X_{0}\sim\mathcal{G}a(a,b),\quad\beta_{0},\beta_{1},\beta_{2}\sim\mathcal{G}a(c,d)
\]
 where $\mathcal{G}a(a,b)$ is a Gamma distribution with mean $a/b$
and $\theta=(\beta_{0:2})\in(\mathbb{R}^{+})^{3}$. This is a GARCH(1,1)
model with an \emph{intractable likelihood}, i.e.~one cannot perform exact parameter inference and has to resort to approximations.

We will consider a comparison of Algorithms \ref{alg:Ntry} and \ref{alg:new}. 
and we will fit the model to data from  the S\&P 500 index from 03/1/11 to 14/02/13 (533 data points).  In the priors, we set $a=c=2$ and $b=d=1/8$, which are not overly informative.
In addition, $s_1=1.5$ and $s_2=0$; see \cite{jasra1} for explanations. We consider $\epsilon\in\{0.01,0.5\}$ and only a noisy ABC approximation of the model.
 The MCMC proposals on the parameters are normal random-walks on the log-scale and for both algorithms we set $N=250$.

In Figure \ref{fig:tracee05} (a) we present the autocorrelation plot of 50000 iterations of both MCMC kernels when $\epsilon=0.5$ (for only the parameter $\beta_0$ - other parameters can be seen in \cite{jasra1}). 
Algorithm \ref{alg:Ntry} took about 0.30 seconds per iteration and Algorithm \ref{alg:new} took about 1.12 seconds per iteration. 
The plot shows that both algorithms appear to mix across the state-space in a very reasonable way. The MCMC procedure associated to Algorithm \ref{alg:new} takes much longer and in this situation does not appear to be required.

In Figure \ref{fig:tracee05} (b)  we can observe the autocorrelation plots from a particular (typical) run when $\epsilon=0.01$. In this case, both algorithms are run for 200000 iterations.
Algorithm \ref{alg:Ntry} took about 0.28 seconds per iteration and Algorithm \ref{alg:new} took about 2.06 seconds per iteration.
Whilst the computational time for Algorithm \ref{alg:new} is considerably more than Algorithm \ref{alg:Ntry}, in the same amount of computation time, it still moves more around the state
space as suggested by Figure \ref{fig:tracee05} (b); algorithm runs of the same length are provided for presentational purposes.
We were unable to tune Algorithm \ref{alg:Ntry}  to make it mix well in this example and with many efforts. Conversely, for Algorithm \ref{alg:new} we expended considerably less effort for very reasonable performance.
The results here suggest that one
should prefer Algorithm \ref{alg:new} only in challenging scenarios, as it can be very expensive in practice.

\begin{figure}[h]
\subfigure[$\beta_0$]{{\includegraphics[width=0.49\textwidth,height=7.5cm]{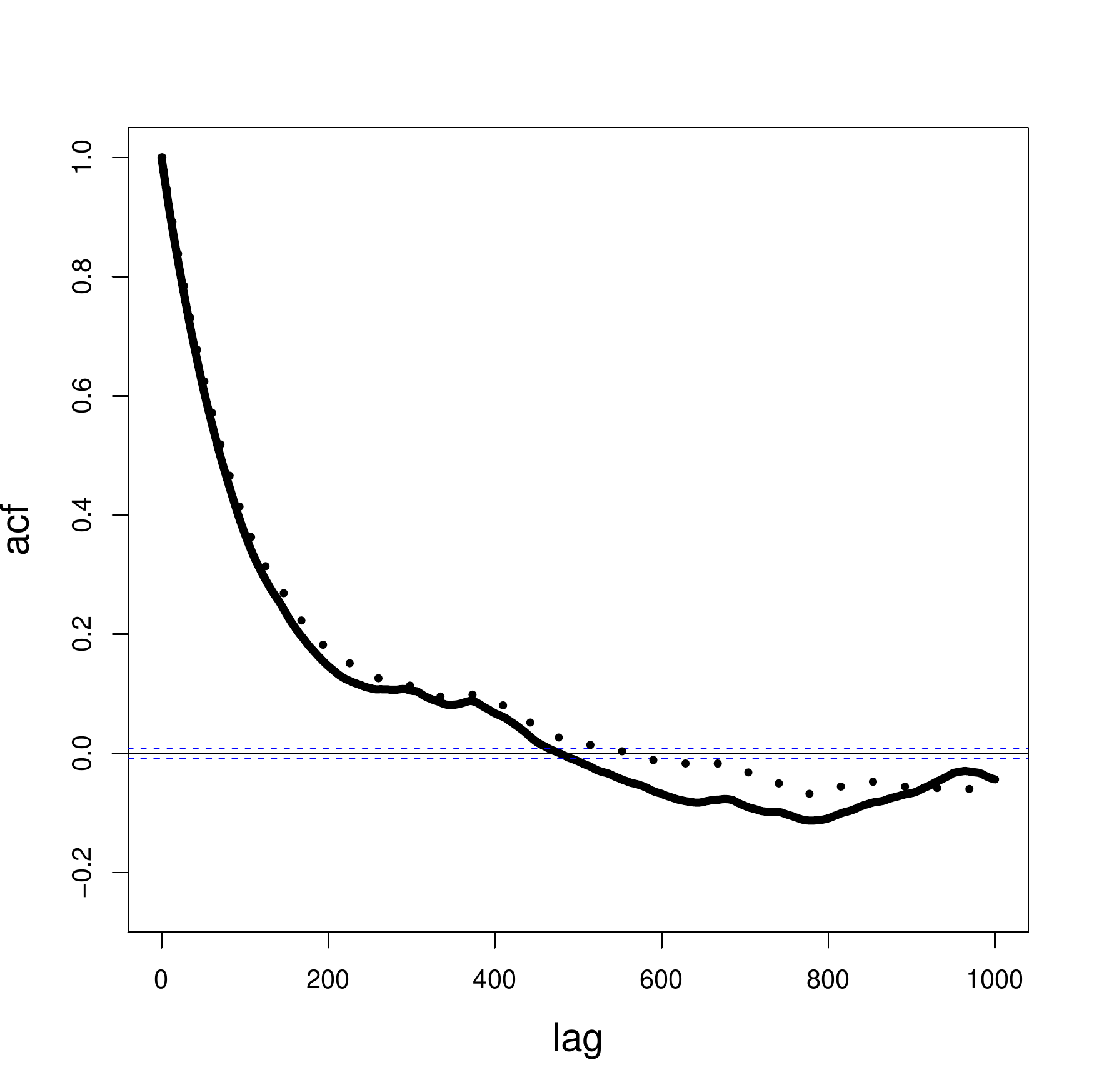}}}
\subfigure[$\beta_0$]{{\includegraphics[width=0.49\textwidth,height=7.5cm]{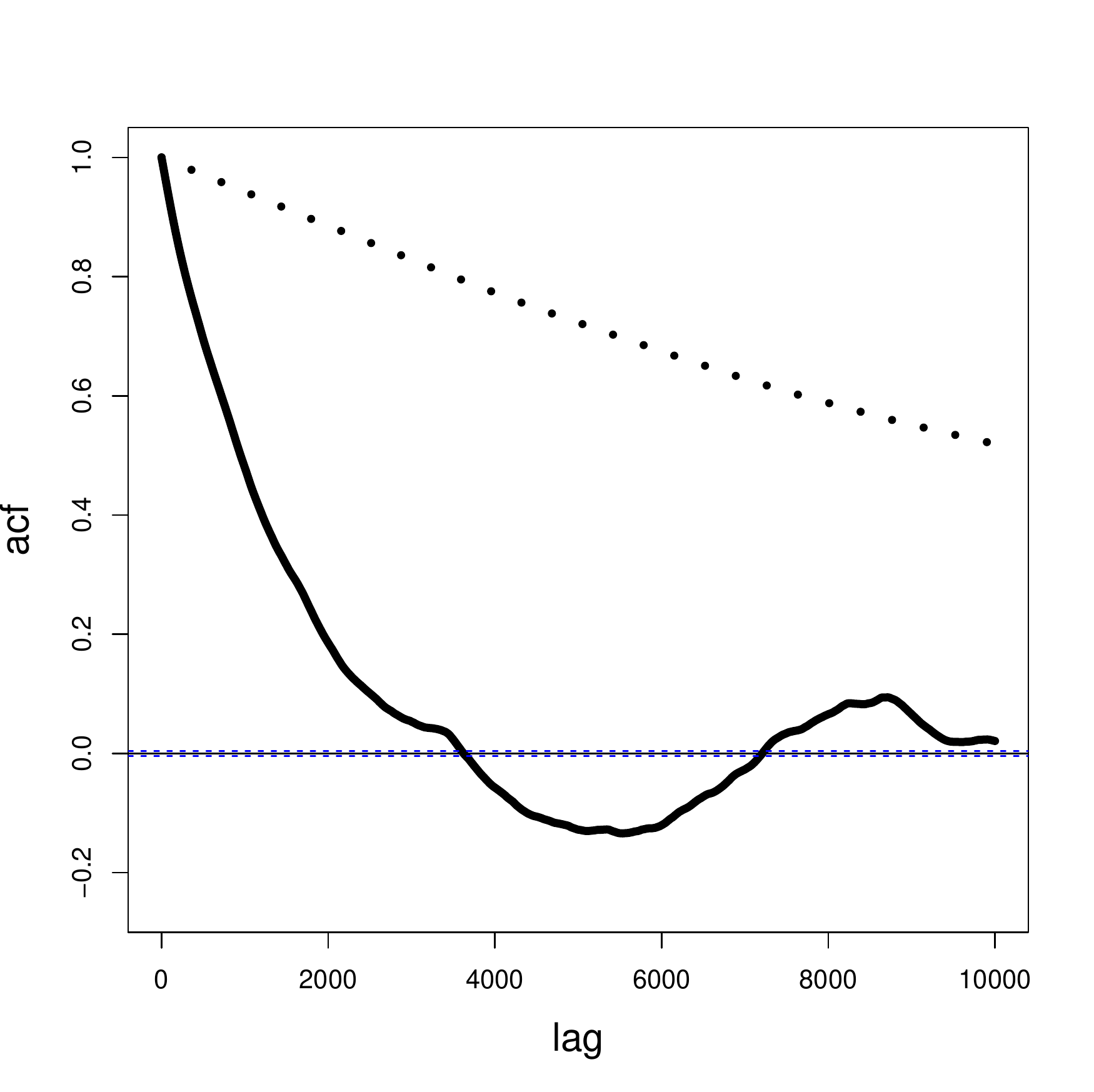}}}
\caption{Autocorrelation Plots for the Sampled Parameter $\beta_0$. In the left plot (a), we run Algorithm \ref{alg:Ntry} (dot) and \ref{alg:new} (full) for 50000 iterations (both with $N=250$)
on the S \& P 500 data, associated to noisy ABC and $\epsilon=0.5$. 
In the right plot we run Algorithm \ref{alg:Ntry} (dot) and \ref{alg:new} (full) for 200000 iterations (both with $N=250$)
on the S \& P 500 data, associated to noisy ABC and $\epsilon=0.01$.
The dotted horizontal lines are a default confidence interval generated by the R package.}
\label{fig:tracee05} 
\end{figure}

\subsubsection{Using Non-Negative Kernels}

If one chooses to use a non-negative kernel in the ABC approximation, then some the issues discussed above will not be so relevant.
In such scenarios one may attempt to use Algorithm \ref{alg:simple} replacing the indicators with non-negative kernels. If this does not work
well, one can resort to more advanced methods such as those in \cite{andrieu}.

\subsection{SMC and MCMC Methods for HMMs}

\subsubsection{SMC Algorithms for Filtering}

We first begin with SMC algorithms which can approximate the ABC filter, that is, when $\theta$ is fixed. As we shall see, these type of algorithms will form a building block of the MCMC
algorithms to be presented in the next section.

We first present a standard SMC algorithm in Algorithm \ref{alg:smc_abc}, which can be found in, for example \cite{calvet,jasra}.
The algorithm will allow one to approximate finite expecations w.r.t.~the ABC filter with density $\pi_{\theta}^{\epsilon}(x_n,u_n|y_{1:n})$ via
$$
\sum_{i=1}^N\frac{W_n^i}{\sum_{j=1}^N W_n^j}\xi(x_n^i,u_n^i)
$$
where $\xi:\mathsf{X}\times\mathsf{Y}\rightarrow\mathbb{R}$ is an integrable function (w.r.t.~the ABC filter). Convergence results as $N$ grows can be found in \cite{delmoral1,delmoral2}.
In Algorithm \ref{alg:smc_abc} an (unbiased (expectation w.r.t.~the randomness in the SMC algorithm) - see \cite{delmoral1,delmoral2}) estimate of $p_{\theta}^{\epsilon}(y_{1:n})=\int_{\mathsf{X}^n}p_{\theta}^{\epsilon}(y_{1:n},x_{1:n})dx_{1:n}$ is 
\begin{equation}
p_{\theta}^{\epsilon,N}(y_{1:n}) = \prod_{i=1}^n\bigg\{\Big(\frac{1}{\int_{\{u\in\mathsf{Y}:\mathsf{d}(u,y_i)<\epsilon\}}du}\Big)\Big(\frac{1}{N}\sum_{j=1}^N W_i^j\Big)\bigg\}.\label{eq:nc_smc}
\end{equation}
This point is useful later on. The algorithm will be used as part of an MCMC proposal, that is, SMC within MCMC.

On inspection of Algorithm \ref{alg:smc_abc}, it can be observed that there is a possibility that all the weights $\{W_n^i\}_{1\leq i \leq N}$  are zero and at which point, the algorithm is said to have collapsed.
This could be alleviated by making $\epsilon$ larger, but then one loses accuracy in the ABC approximation; before doing this, one may want to try improving the SMC algorithm.
The property of the algorithm collapsing is not desirable, especially in the context in which we will use it - the idea of using it as part of a proposal in MCMC. That is to say, that the simulation of Algorithm \ref{alg:smc_abc} will be
the major contributor to the computational cost of an associated MCMC algorithm and, if the SMC approach collapses (especially close to time $n$) one will have spent considerable effort on a proposal which will be rejected.

One algorithm which has been proposed in the literature and which can potentially deal with this problem is the alive particle filter (see \cite{amerin,jasra2}, see also \cite{legland1}). This idea is related
to Algorithm \ref{alg:new} in that it uses a random amount of effort to ensure that there are $N$ samples alive. The alive particle filter is presented in Algorithm \ref{alg:alive_smc}. It should be noted
that this algorithm only retain $N-1$ of the alive samples - this allows the following property: an (unbiased (expectation w.r.t.~the randomness in the SMC algorithm) - see \cite{amerin,jasra2}) estimate of $p_{\theta}^{\epsilon}(y_{1:n})$ is 
\begin{equation}
p_{\theta}^{\epsilon,N}(y_{1:n}) = \prod_{i=1}^n\frac{N-1}{(m_i-1)(\int_{\{u\in\mathsf{Y}:\mathsf{d}(u,y_i)<\epsilon\}}du)}.\label{eq:nc_alive}
\end{equation}
The unbiased property will be useful for the MCMC algorithms in the next Section. It should be noted that one can also approximate the ABC filter just as above and convergence results can be found in \cite{jasra2} (note that
the result in \cite{legland1} are for different estimates). This algorithm has been further extended in \cite{persing} and we refer the reader to that article for details.

\begin{algorithm}
\begin{enumerate}
\item {Intialization: For $i\in\{1,\dots,N\}$, sample  $x_1^i$ from a distribution with density $q_{\theta}(x_1)$ and $u_1^i|x_1^i$ from the likelihood. Compute the weight $W_1^i = \mathbb{I}_{\{u\in\mathsf{Y}:\mathsf{d}(u,y_1)<\epsilon\}}(u_1^i)p_{\theta}(x_1^i|x_0)/q_{\theta}(x_1^i)$. Set $n=1$}
\item{Resampling:  Resample the particles (setting $W_n^i=1$). Denote the resulting particles $\{(\hat{x}_n^i,\hat{u}_n^i)\}_{1\leq i \leq N}$.}
\item{Sampling: For $i\in\{1,\dots,N\}$, sample  $x_{n+1}^i|\hat{x}_{n}^i$ from a distribution with density $q_{\theta}(x_{n+1}|\hat{x}_n^i)$ and $u_{n+1}^i|x_{n+1}^i$ from the likelihood. Compute the weight $W_{n+1}^i = \mathbb{I}_{\{u\in\mathsf{Y}:\mathsf{d}(u,y_{n+1})<\epsilon\}}(u_{n+1}^i)p_{\theta}(x_{n+1}^i|\hat{x}_{n}^i)/q_{\theta}(x_{n+1}^i|\hat{x}_{n}^i)$. Set $n=n+1$ and go to 2.}
\end{enumerate}
\caption{SMC for ABC.}
\label{alg:smc_abc}
\end{algorithm}

\begin{algorithm}
\begin{enumerate}
\item {Intialization: Sample  $x_1^1,\dots$ from a distribution with density $q_{\theta}(x_1)$ and $u_1^1|x_1^1,\dots$ from the likelihood until there are $N$ samples
with $W_1^i>0$, $W_1^i = \mathbb{I}_{\{u\in\mathsf{Y}:\mathsf{d}(u,y_1)<\epsilon\}}(u_1^i)p_{\theta}(x_1^i|x_0)/q_{\theta}(x_1^i)$, record the number of samples required to do this, call it $m_1$. Set $n=1$ and discard all samples except the first $N-1$ with non-zero weight (call these the alive particle).}
\item{Iteration: Resample (amongst the alive particles) $\hat{x}_n^1,\dots$ according to the weights $W_n^i$ then sample  $x_{n+1}^1|\hat{x}_{n}^1,\dots$ from a distribution with density $q_{\theta}(x_{n+1}|\hat{x}_{n}^1),\dots$ and $u_{n+1}^1|x_{n+1}^1,\dots$ from the likelihood
until there are $N$ samples with $W_{n+1}^i>0$, $W_{n+1}^i =\mathbb{I}_{\{u\in\mathsf{Y}:\mathsf{d}(u,y_{n+1})<\epsilon\}}(u_{n+1}^i)p_{\theta}(x_{n+1}^i|\hat{x}_{n}^i)/q_{\theta}(x_{n+1}^i|\hat{x}_{n}^i)$, record the number of samples required to do this, call it $m_{n+1}$.
Set $n=n+1$ and discard all samples except the first $N-1$ with non-zero weight and return to the start of 2..}
\end{enumerate}
\caption{Alive SMC.}
\label{alg:alive_smc}
\end{algorithm}

To end the section we note that if smoothing is of interest, then any of the standard SMC approaches can be adopted such as \cite{briers,dds1}; see \cite{martin} for some implementations.

\subsubsection{Particle MCMC Algorithms}

We now consider trying to sample from the ABC approximation of the posterior on the parameters, i.e.~$\pi^{\epsilon}(\theta|y_{1:n})\propto p_{\theta}^{\epsilon}(y_{1:n})\pi(\theta)$;
as this is not possible, we will construct algorithms on extended state-spaces, such as in Section \ref{sec:mcmc_for_iid}. It should be noted that algorithms which attempt to
sample the ABC posterior \eqref{eq:abc_hmm_post}
using `standard' MCMC methods, often do not work well; we direct the reader to \cite{andrieu} for an explanation.
We present two basic particle marginal Metropolis-Hastings algorithms (PMMH) \cite{andrieu}, which can work well in practice.

The first MCMC algorithm that we present is in Algorithm \ref{alg:pmmh}. This algorithm uses Algorithm \ref{alg:smc_abc} to propose new points in the state-space;
as remarked above, due to the possiblity that the SMC may collapse this algorithm may not be entirely satisfactory. 
We do not give details of the target density (see \cite{andrieu}) but note that $\pi^{\epsilon}(\theta|y_{1:n})$ is a marginal (which essentially follows because the estimate $p_{\theta}^{\epsilon,N}(y_{1:n})$ is unbiased)
 and one of the key issues
is the relative variance of $p_{\theta}^{\epsilon,N}(y_{1:n})$ w.r.t.~the SMC algorithm that is simulated.
In order to select $N$ one can appeal to theoretical
results in SMC (\cite{cerou1}) which, as noted by  \cite{andrieu} say that one should select $N=\mathcal{O}(n)$ to control the relative variance of the estimate $p_{\theta}^{\epsilon,N}(y_{1:n})$
which appears in the acceptance probability. More precise results for selecting $N$ can be found in \cite{pitt,sherlock}. Note that to fit the collapsed representation of a HMM (as in Section \ref{sec:collapse}) one
can use a PMMH algorithm very similar to Algorithm \ref{alg:pmmh}; see \cite{murray} for a PMMH algorithm for a collapsed representation of a HMM.

\begin{algorithm}
\begin{itemize}
\item {Sample $\theta'|\theta$ from a proposal $Q(\cdot|\theta)$ with
density $q(\theta'|\theta)$.} 
\item{Run Algorithm \ref{alg:smc_abc} with parameter $\theta'$ and record the estimate $p_{\theta'}^{\epsilon,N}(y_{1:n})$ in \eqref{eq:nc_smc}.}
\item{Accept $\theta'$ with probability:
$$
1\wedge \frac{p_{\theta'}^{\epsilon,N}(y_{1:n})}{p_{\theta}^{\epsilon,N}(y_{1:n})} \times \frac{\pi(\theta')q(\theta|\theta')}{\pi(\theta)q(\theta'|\theta)}.
$$
}
\end{itemize}
\caption{Particle Marginal Metropolis-Hastings, with standard SMC.}
\label{alg:pmmh}
\end{algorithm}

The second MCMC algorithm we present is in Algorithm \ref{alg:pmmh_alive}. This MCMC kernel uses Algorithm \ref{alg:alive_smc} as part of its proposal mechanism.
Essentially this algorithm removes the issue of the SMC collapsing, but the associated cost of each proposal may be significantly higher than in Algorithm \ref{alg:pmmh}.
As for the previous MCMC algorithm, we do not give details of the target (see \cite{jasra2}) but one of the key issues why $\pi^{\epsilon}(\theta|y_{1:n})$ is a marginal
is that the estimate $p_{\theta}^{\epsilon,N}(y_{1:n})$ is unbiased. The remarks for setting $N$ are the same as for Algorithm \ref{alg:pmmh}, although, we note that
(under assumptions) the work in \cite{jasra2} indicates that the relative variance of \eqref{eq:nc_alive} is less than that of \eqref{eq:nc_smc} and that potentially Algorithm
\ref{alg:pmmh_alive} is more efficient than Algorithm \ref{alg:pmmh}; we discuss further in the next section. It should be noted that for both Algorithms \ref{alg:pmmh} and \ref{alg:pmmh_alive}
inference for the state-sequence $x_{1:n}$ can also be performed and we refer to \cite{andrieu} for details. We also note that a version of Algorithm \ref{alg:pmmh_alive} has been adopted
in \cite{drov} for time series models, when $\mathsf{Y}$ is countably-infinite; in such scenarios the data can be matched leading to exact inference - we refer the reader to \cite{drov} for more details.

\begin{algorithm}
\begin{itemize}
\item {Sample $\theta'|\theta$ from a proposal $Q(\cdot|\theta)$ with
density $q(\theta'|\theta)$.} 
\item{Run Algorithm \ref{alg:alive_smc} with parameter $\theta'$ and record the estimate $p_{\theta'}^{\epsilon,N}(y_{1:n})$ in \eqref{eq:nc_alive}.}
\item{Accept $\theta'$ with probability:
$$
1\wedge \frac{p_{\theta'}^{\epsilon,N}(y_{1:n})}{p_{\theta}^{\epsilon,N}(y_{1:n})} \times \frac{\pi(\theta')q(\theta|\theta')}{\pi(\theta)q(\theta'|\theta)}.
$$
}
\end{itemize}
\caption{Particle Marginal Metropolis-Hastings, with Alive SMC.}
\label{alg:pmmh_alive}
\end{algorithm}

The section is concluded by remarking that there are more advanced algorithms in the literature. Examples include the approach in \cite{persing} and use of the particle Gibbs sampler
(\cite{andrieu}) for example using the alive particle filter as in \cite{zhang} (who also considers the collapsed representation for doubly intractable HMMs - the associated PMMH algorithm is very similar
to Algorithm \ref{alg:pmmh}).

\subsubsection{Simulations}

We consider the following HMM, for $n\geq 1$
\begin{align*}
Y_n =& \phi_n \beta \exp (X_n) \\
X_n =& \varepsilon X_{n-1} + \eta_n
\end{align*}
where $x_0=0$, $\phi_n\sim\mathcal{S}t(0,s_1,s_2,s_3)$ (a stable distribution with location
parameter 0, scale $s_1$ asymmetry parameter $s_2$ and skewness parameter $s_3$) and $\eta_n \sim \mathcal{N}(0,c)$. 
We set $\theta=(\beta,c,\varepsilon)$, with priors 
$c\sim \mathcal{IG}(2,1/100)$, $\varepsilon\sim \mathcal{IG}(2,1/50)$
($\mathcal{IG}(a,b)$ is an inverse Gamma distribution with mode $b/(a+1)$)
and $\beta\sim \mathcal{N}(0,10)$. Note that the inverse Gamma distributions have infinite variance.
We again consider the daily index of the S \& P 500 index between 03/01/2011 $-$ 14/02/2013 (533 data points). 

We consider two scenarios to compare Algorithms \ref{alg:pmmh} and \ref{alg:pmmh_alive}. 
In the first situation we set $s_2=1.75$ and in the second, $s_2=1.2$, with $s_1=s_3=1$ in both situations. In the first case, we make $\epsilon$ not too large
and in the second, $\epsilon$ is significantly reduced. The noisy ABC approximation is used.
The SMC and alive SMC algorithms used the transition dynamics of the $\{X_n\}_{n\geq 1}$ as the $q_{\theta}$ densities in Algorithms \ref{alg:smc_abc} and \ref{alg:alive_smc}.
Both algorithms are run for about the same computational time, such that Algorithm \ref{alg:pmmh_alive} is run for 20000 iterations. The parameters are initialized with draws
from the priors. The proposal on $\beta$ is a normal random walk
and for $(c,\varepsilon)$ a gamma proposal centered at the current point with proposal variance scaled to obtain reasonable acceptance rates. We consider $N=1000$ and for  Algorithm \ref{alg:pmmh_alive} this value is lower to allow the same computational time.

Our results are presented in Figure \ref{fig:pmcmc}.
In Figures \ref{fig:pmcmc} (a)-(b) we can see the output in the case that $s_2=1.75$ (only for $\beta$, more results are in \cite{jasra2}). 
It appears that both algorithms perform very well; the acceptance rates were around 0.25 for each case.
For this scenario one would prefer Algorithm \ref{alg:pmmh} as the algorithmic performance is very good, with a removal of a random computational
cost per iteration.

In Figures \ref{fig:pmcmc} (c)-(d) the output when $s_2=1.2$ is displayed. In Figure \ref{fig:pmcmc} (c)
 we can see that Algorithm \ref{alg:pmmh} performs very badly, barely moving across the parameter space (despite significant efforts to tune the algorithm),
whereas the new PMMH algorithm has very reasonable performance (Figure \ref{fig:pmcmc} (d)). In this case, $\epsilon$ is very small,
and Algorithm \ref{alg:smc_abc} collapses very often, which leads to the undesirable performance displayed. As in Section \ref{sec:simos_iid}
the results here suggest that Algorithm \ref{alg:pmmh_alive} might be preferred in difficult sampling scenarios, but in simple cases it does not seem to be required.

\begin{figure}[h]
\subfigure[$\beta$, Algorithm \ref{alg:pmmh}, $s_2=1.75$]{{\includegraphics[width=0.49\textwidth,height=7.5cm]{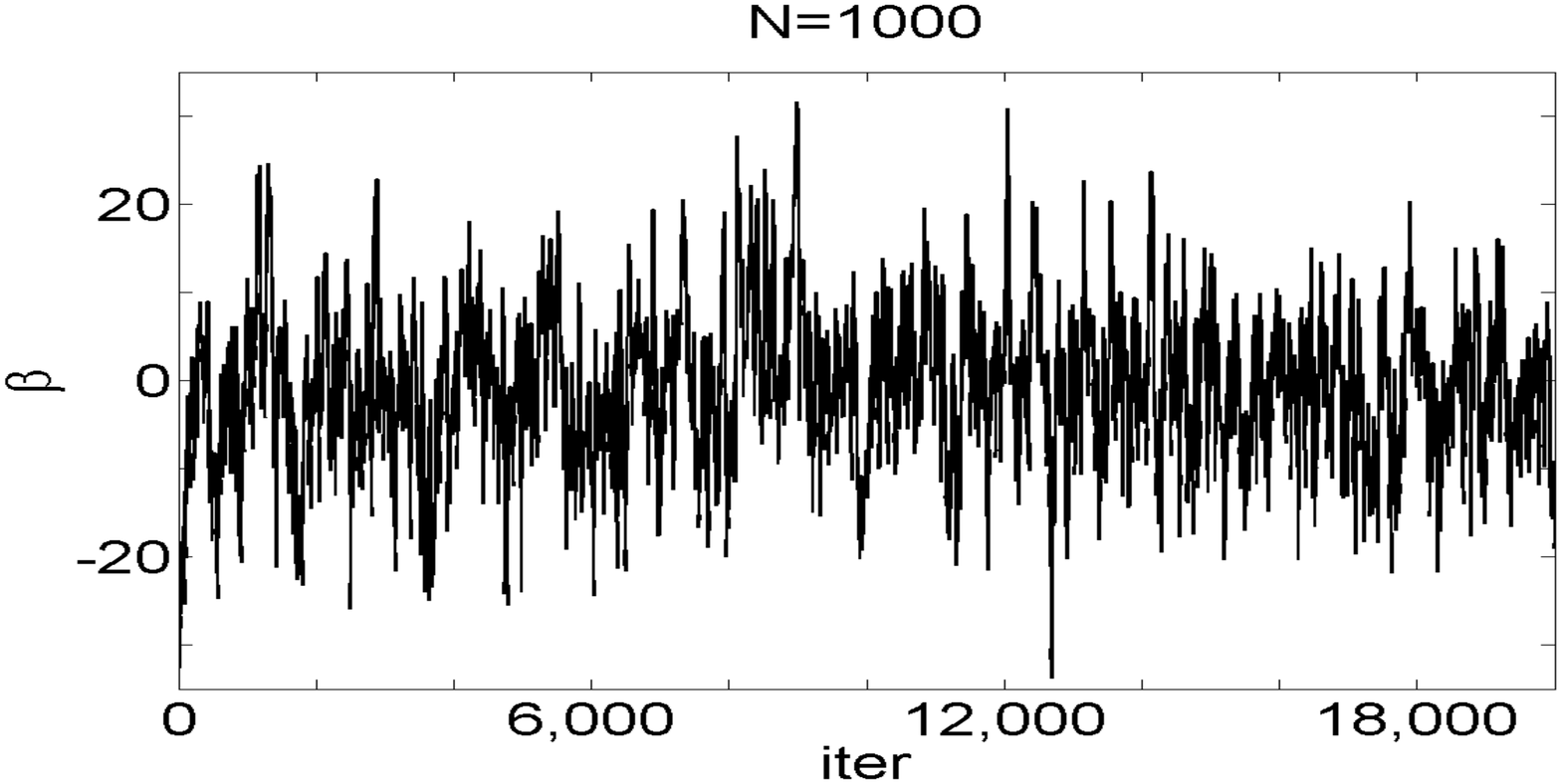}}}
\subfigure[$\beta$, Algorithm \ref{alg:pmmh_alive}, $s_2=1.75$]{{\includegraphics[width=0.49\textwidth,height=7.5cm]{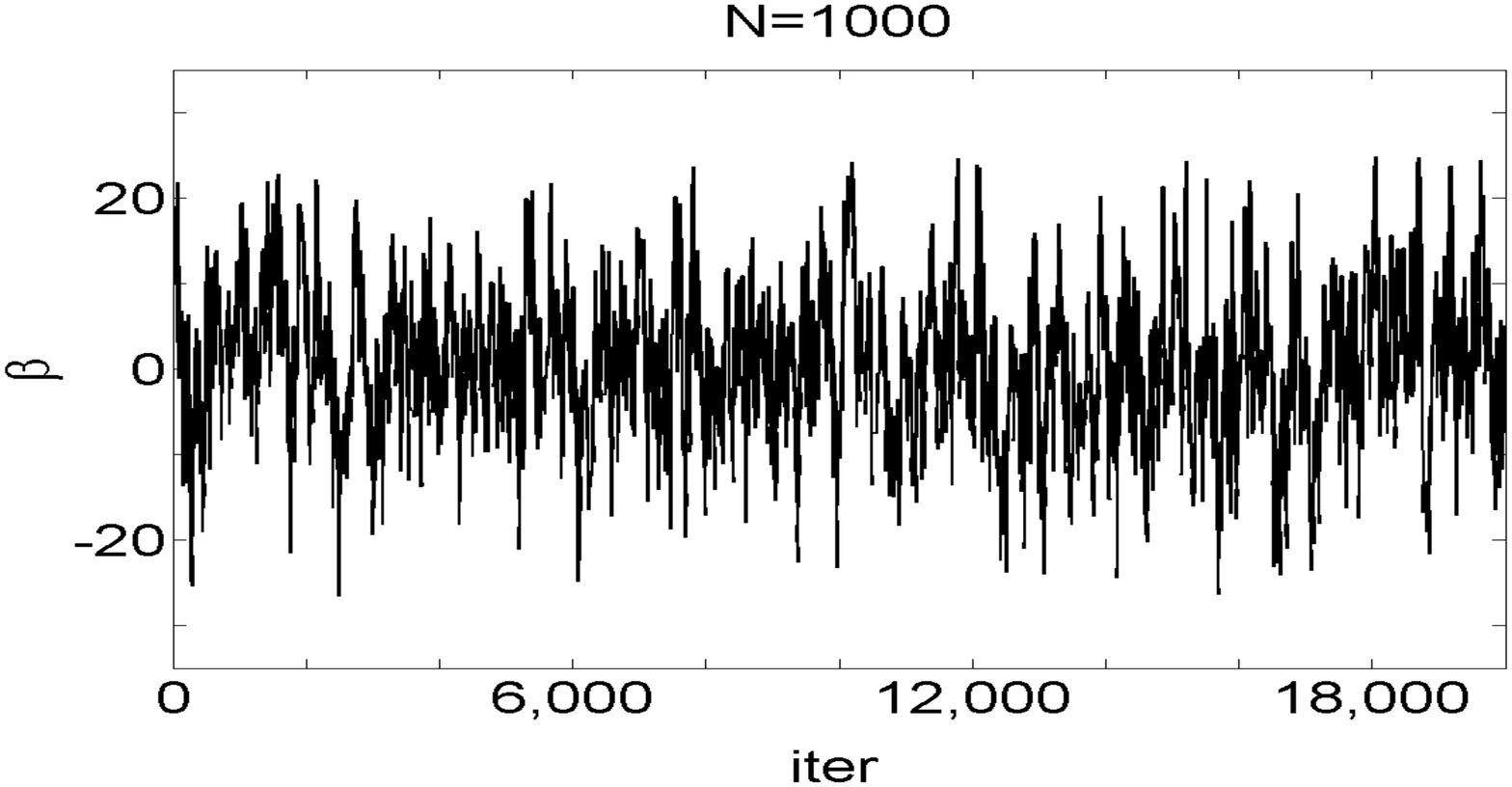}}}
\subfigure[$\beta$, Algorithm \ref{alg:pmmh}, $s_2=1.2$]{{\includegraphics[width=0.49\textwidth,height=7.5cm]{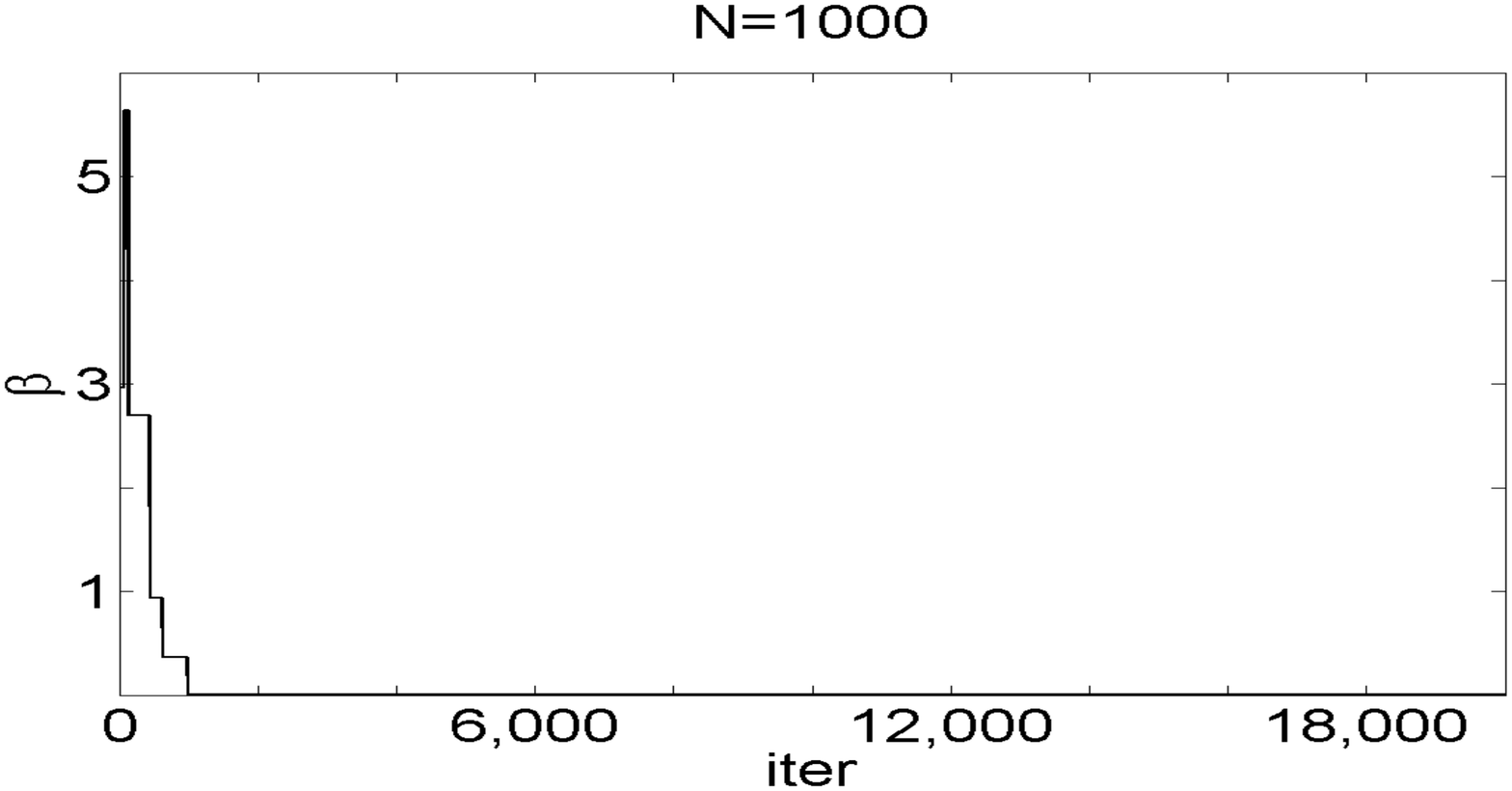}}}
\subfigure[$\beta$, Algorithm \ref{alg:pmmh_alive}, $s_2=1.2$]{{\includegraphics[width=0.49\textwidth,height=7.5cm]{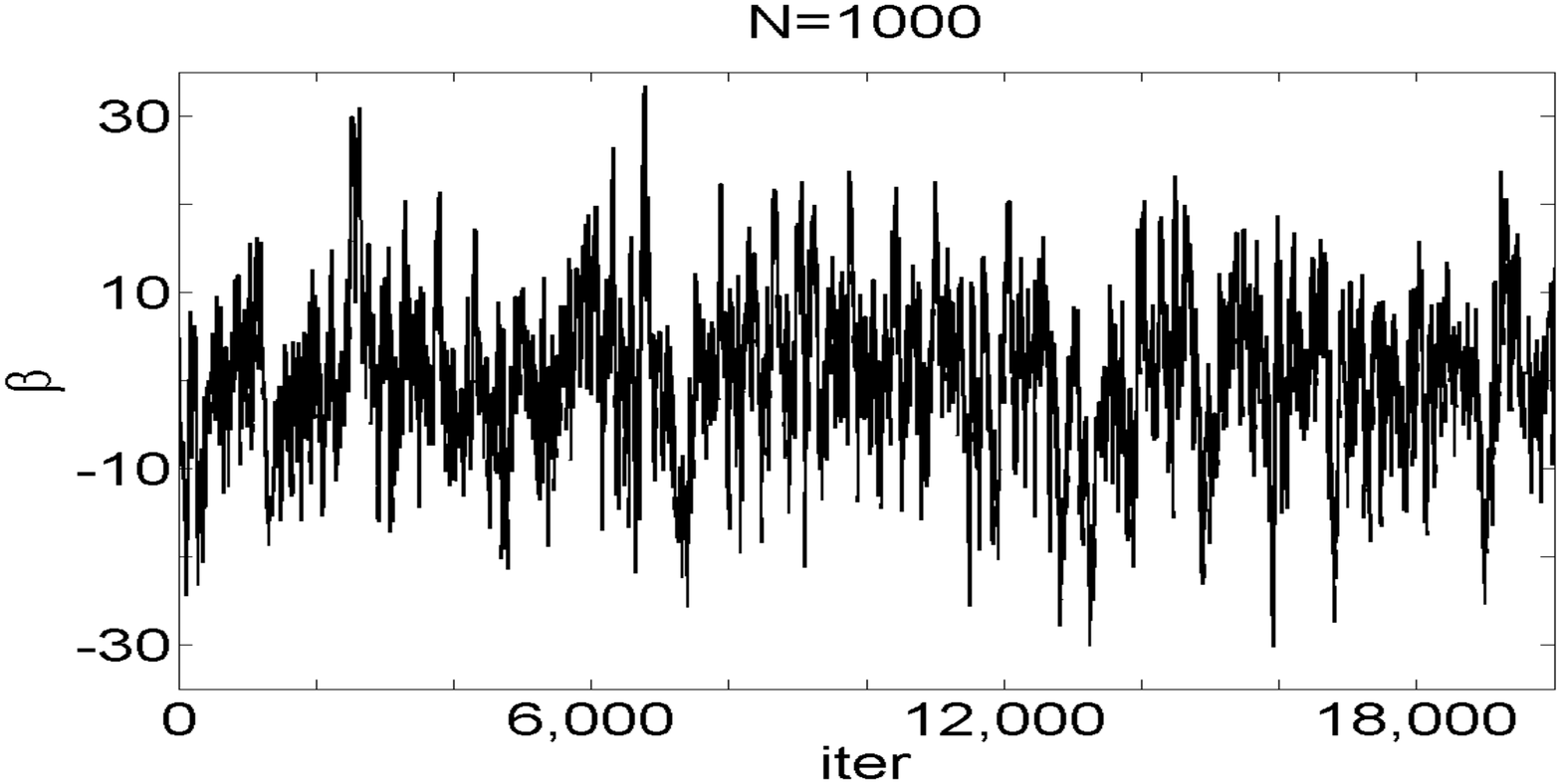}}}
\caption{Trace Plots for the Sampled Parameter $\beta$. We run Algorithms \ref{alg:pmmh} and \ref{alg:pmmh_alive} with $N=1000$ (for Algorithm \ref{alg:pmmh_alive} this value is lower to allow a similar computational time) for two different values of $s_2$. Algorithm \ref{alg:pmmh_alive} is run for 20000 iterations.}
\label{fig:pmcmc} 
\end{figure}

\subsubsection{Using Non-Negative Kernels}

As for i.i.d.~models if one chooses to use a non-negative kernel in the ABC approximation, then some the issues discussed above will not be so relevant.
A starting point for MCMC algorithms would be Algorithm \ref{alg:pmmh} which uses Algorithm \ref{alg:smc_abc} except with the indicator replaced with the non-negative kernel.

\subsection{Expectation-Propagation}

The computational methods discussed so far focus on exact computation for the ABC approximation. The methods are certainly not simple to understand and often take a considerable amount of effort in both coding
as well as time to run. An idea put fourth in \cite{chopin1}, for exactly the class of models (and ABC approximations) considered in this article, is to try to infer the ABC posterior using \emph{approximate}
computational methods. The authors also state that the approach can be used as an initial investigation of data, instead of possibly the `final' data analysis. The idea is to use the expectation-propagation algorithm
of \cite{minka}. The authors successfully implement the methodology for a wide variety of statistical models. They also compare the methods to the PMMH methods investigated in this article. Considerable gains in computational time
are reported, with some loss of accuracy in inference. This method seems to be very promising and an interesting alternative to using exact computational methods. The fact remains, however, that a theoretical study of the increased
loss in accuracy in inference may be more challenging than ABC with exact computation - although the latter task is still not straight-forward.

\subsection{Online Parameter Inference}

This article has focussed upon time-series models, but the computational techniques discussed thus far have focussed on batch parameter inference. We now briefly discuss
methodologies that have been used to perform online parameter inference for ABC, that is as data arrive sequentially.

In the context of i.i.d.~and ODTS models one techniuqe that has been adopted in \cite{ehrlich1} is the SMC sampler method of \cite{chopin,delm:06}. This approach allows one to estimate the posterior density on $\theta$
as data arrive. The method is based on MCMC and typically the computational effort per time step grows with the time parameter, but the overall effort is only polynomial in $n$ - $\mathcal{O}(n^2)$.

For HMMs estimating the posterior density on $\theta$ is a notoriously difficult computational problem (for SMC methods), even when the likelihood is tractable; see \cite{doucet}. One possible approach is in \cite{chopin_smc}, but again the computational
cost per time step (which is polynomial in $n$) increases with the time parameter. Online parameter estimation methods (i.e.~whose cost does per-time step does not increase with the time parameter), based upon maximum likelihood can be found in \cite{ehrlich,yildrlim}; note that
\cite{yildrlim} has also been used for i.i.d.~data. \cite[pp.~435]{fearnhead} also considers a type of online estimation procedure for Lotka-Voltera models.

\section{Discussion}\label{sec:summary}

In this article we have discussed both ABC approximations and computational methods for a class of time series models found in many applications.
There are a number of important and interesting research directions, some of which we now list.

From a mathematical perspective there are several aspects which could be investigated. Firstly, many of the consistency and asymptotic normality results have been given for specific models.
However, one suspects that general theorems (for both classical and Bayesian consistency) for the entire class of models could be proved. In addition, in specific scenarios that have been discussed,
the mathematical assumptions could be significantly relaxed. This discussion concerns only the approximation, and as such suggests that one wishes $\epsilon$ to be as small as possible. As has been
remarked in Section \ref{sec:comp}, making $\epsilon$ arbitrarily small will make even the most sophisticated computational algorithms collapse. So there is a need to both analyze biases (which are deterministic)
and the computational errors (which are stochastic) specifically for problems in ABC - this could allow one to find (as done in \cite{fearnhead}) an optimal $\epsilon$ in specific scenarios 
to facilitate a trade off in the errors (a point also made in \cite{marin}). To further expand, the approach advocated by this author (and others see \cite{jasra3,marin,martin}), is to decompose
the error in approximation as follows; let $\pi(\xi)$ be the expectation of a real-valued $\pi-$integrable function $\xi$ under the true posterior, $\pi^{\epsilon}(\xi)$ the ABC approximation of it and $\pi^{\epsilon,N}(\xi)$ a Monte Carlo based (e.g.~via MCMC) approximation
of $\pi^{\epsilon}(\xi)$, then (for example) considering the $\mathbb{L}_1-$error, one has
$$
\mathbb{E}^N[|\pi^{\epsilon,N}(\xi)-\pi(\xi)|] \leq \mathbb{E}^N[|\pi^{\epsilon,N}(\xi)-\pi^{\epsilon}(\xi)|] + |\pi^{\epsilon}(\xi)-\pi(\xi)|
$$
where $\mathbb{E}^N$ is the expectation w.r.t.~the Monte Carlo randomness.
The idea is to then to separately deal with the two errors on the R.H.S.~the first of which is the Monte Carlo error and the second the bias (for finite number of data even noisy ABC will have a bias).
Some work has been done in \cite{martin} but under very restrictive mathematical assumptions
and less restrictive (and more complete) results are in \cite{calvet}, and in the cases of rejection sampling and regression adjustment \cite{biau,blum}, but some additional work is needed in general.

From a computational perspective, whilst we have presented some of the most up-to date methods in the literature, there is still a need to improve upon this. For example, one cannot guarantee that the alive particle
filter will fix all the issues with standard SMC for HMMs. Perhaps the direction to follow is as in \cite{chopin1}, where the notion of being approximate in computation also, is advocated, instead of a collection of even
more sophisticated tools which can be difficult to understand for statistical practitioners. Some work on statistical software for the class of time series models in this paper also appears to be required.

From an inferential perspective there are many issues to be considered. Firstly this paper has only considered, in the main, sub-classes of the time series models of interest. An increased consideration of other models
(such as bilinear time series models) would be of interest (mathematically, computationally and inferentially). Secondly, the article has only considered models for which the summary statistics are the data themselves.
As the demands of applications in `big data' and high-dimensions increases, these ideas will become less relevant and a consideration of this issue must be studied. 
As noted by \cite{chopin1}, the use of `local' statistics (for each data-point, in the case of this paper) may be a more manageable problem than the use of global statistics as often adopted in the literature.
Thirdly, the article has not considered the aspect of model
selection (for example for HMMs when $\mathsf{X}=\{1,\dots,k\}$ and one wishes to choose $k$), which is very challenging in the context of ABC \cite{robert2}. For the ABC approximation with identity summary statistics (see again \cite{robert2}), one
may suspect that this is just an exercise in application of existing computational methods (e.g.~\cite{zhou}) but in the scenario where summary statistics are adopted this problem will be far more challenging.

\subsubsection*{Acknowledgements}

The author was supported by Singapore MOE grants R-155-000-119-133 and R-155-000-143-112. The author is also affiliated with the risk management institute at the National University of Singapore.
The author thanks Pierre Del Moral, Kody Law and Andrew Stuart for conversations on this work as well as collaborators Sumeetpal Singh, Nikolas Kantas and David Nott for many conversations on ABC. In addition thanks to former
PhD students Zhang and Ehrlich for access to computer code.

\appendix

\section{Proofs}

\subsection{Convergence of ABC Posterior}\label{sec:prfs}

\begin{hypA}\label{hyp:1}
For $n\geq 1$, $y_{1:n}$   are fixed and the posterior exists
$$
\int_{\Theta\times\mathsf{X}^n}  p_{\theta}(y_{1:n}|x_{1:n})p_{\theta}(x_{1:n})\pi(\theta) dx_{1:n}d\theta < +\infty.
$$
\end{hypA}

\begin{proof}[Proof of Proposition \ref{prop:ep_conv}]
The proof is as follows: we will show that for any fixed $x_{1:n}\in\mathsf{X}^n$, $\theta\in\Theta$ that
\begin{equation}
\lim_{\epsilon \rightarrow 0} p_{\theta}^{\epsilon}(y_{1:n}|x_{1:n}) = p_{\theta}(y_{1:n}|x_{1:n})\label{eq:prf1}.
\end{equation}
where 
$$
 p_{\theta}^{\epsilon}(y_{1:n}|x_{1:n}) = \Big(\prod_{i=1}^n 
p_{\theta}^{\epsilon}(y_i|y_{1:i-1},x_i)
\Big) \quad\quad 
 p_{\theta}(y_{1:n}|x_{1:n}) = \Big(\prod_{i=1}^n 
p_{\theta}(y_i|y_{1:i-1},x_i)
\Big).
$$
Once this is achieved, the result is proved by the dominated convergence theorem on noting that one has, via (A\ref{hyp:1})
$$
\lim_{\epsilon\rightarrow 0} \int_{\Theta\times\mathsf{X}^n}\xi(\theta,x_{1:n}) p_{\theta}^{\epsilon}(y_{1:n}|x_{1:n})p_{\theta}(x_{1:n})\pi(\theta) dx_{1:n}d\theta = 
$$
$$
\int_{\Theta\times\mathsf{X}^n}\xi(\theta,x_{1:n}) p_{\theta}(y_{1:n}|x_{1:n})p_{\theta}(x_{1:n})\pi(\theta) dx_{1:n}d\theta
$$
so one can treat the numerators and denominators in Bayes theorem seperately to conclude the result. 

We now proceed to prove \eqref{eq:prf1}, which is done by induction on $n$. Let $n=1$, then we have
$$
p_{\theta}^{\epsilon}(y_1|x_1) - p_{\theta}(y_1|x_1) = 
\frac{1}{\int_{\{u\in\mathsf{Y}:\mathsf{d}(u,y_1)<\epsilon\}}du}\int_{\{u\in\mathsf{Y}:\mathsf{d}(u,y_1)<\epsilon\}}[p_{\theta}(u|x_1) - p_{\theta}(y_1|x_1)]du
$$
this will converge to zero as $\epsilon\rightarrow 0$ by Lebesgue's differentiation theorem (e.g.~\cite{wheedon}). Now assuming the result for $n-1$ we have
$$
 p_{\theta}^{\epsilon}(y_{1:n}|x_{1:n}) - p_{\theta}(y_{1:n}|x_{1:n}) = 
$$
$$
p_{\theta}^{\epsilon}(y_{1:n-1}|x_{1:n-1})[p_{\theta}^{\epsilon}(y_n|y_{1:n-1})  -p_{\theta}(y_n|y_{1:n-1})] +
p_{\theta}(y_n|y_{1:n-1})[p_{\theta}^{\epsilon}(y_{1:n-1}|x_{1:n-1}) - p_{\theta}(y_{1:n-1}|x_{1:n-1})].
$$
By the induction hypothesis, the second term on the R.H.S.~will converge to zero as $\epsilon\rightarrow 0$. The first term on the R.H.S.~will
converge to zero as the term $p_{\theta}^{\epsilon}(y_{1:n-1}|x_{1:n-1})$ converges (by the induction hypothesis) and the term
$[p_{\theta}^{\epsilon}(y_n|y_{1:n-1})  -p_{\theta}(y_n|y_{1:n-1})]$ goes to zero by almost the same argument for intialization of the induction.
This concludes the proof.
\end{proof}

\subsection{Consistency of Noisy ABC}\label{sec:prfs_abc}

\begin{hypA}\label{hyp:2}
\begin{enumerate}
\item{$\Theta$ is compact.}
\item{$p_{\theta}(y)$ is uniformly continuous (i.e.~over $y$) in $\theta$.}
\item{$p_{\theta}(y)=p_{\theta^{\star}}(y)$ if and only if $\theta=\theta^{\star}$.}
\item{There exist a $C<+\infty$ such that $\sup_{\theta,y}\log(p_{\theta}(y))\leq C$.}
\item{$\theta^{\star,\epsilon}$ exists and is unique.}
\end{enumerate}
\end{hypA}

We note that the assumptions are not intended to be weak, simply illustrative, to keep the proof to a short length.

\begin{proof}[Proof of Proposition \ref{prop:abc_ok}]
We prove $\widetilde{\theta}_n^{\epsilon}\stackrel{a.s.}{\rightarrow}\theta^{\star}$~first; the proof of first statement will follow. In the case of 2.~the $Z_i\stackrel{i.i.d.}{\sim}G_{\theta^{\star}}^{\epsilon}(\cdot)$, so one need only check that the conditions required for
consistency of the MLE are satisfied. We list 4 assumptions often used in the literature (see \cite{ferguson}):
\begin{enumerate}
\item{$\Theta$ is compact.}
\item{$p_{\theta}^{\epsilon}(y)$ is continuous in $\theta$ for every $y\in\mathsf{Y}$.}
\item{$p_{\theta}^{\epsilon}(y)=p_{\theta^{\star}}^{\epsilon}(y)$ if and only if $\theta=\theta^{\star}$.}
\item{There exist a $K(y)$, with $\int_{\mathsf{Y}}K(y)p_{\theta^{\star}}^{\epsilon}(y)dy<+\infty$ such that $\sup_{\theta}\log(p_{\theta}^{\epsilon}(y))\leq K(y)$.}
\end{enumerate}
If one can verify 1.-4.~above then one has $\widetilde{\theta}_n^{\epsilon}\stackrel{a.s.}{\rightarrow}\theta^{\star}$. 1.~holds by assumption. For 2.~let $\kappa>0$,
then one can find a $\delta>0$ such that for $|\theta_1-\theta_2|<\delta$ we have
$$
|p_{\theta_1}^{\epsilon}(y)-p_{\theta_2}^{\epsilon}(y)| = \frac{1}{\int_{\{u\in\mathsf{Y}:\mathsf{d}(u,y)<\epsilon\}}du}
\Big|\int_{\{u\in\mathsf{Y}:\mathsf{d}(u,y)<\epsilon\}}p_{\theta_1}(u)-p_{\theta_2}(u)du\Big| <\kappa
$$
where we have used (A\ref{hyp:2}-2); thus 2.~holds. For 3.~we have
$$
p_{\theta}^{\epsilon}(y) = \frac{1}{\int_{\{u\in\mathsf{Y}:\mathsf{d}(u,y)<\epsilon\}}du}\int_{\{u\in\mathsf{Y}:\mathsf{d}(u,y)<\epsilon\}}p_{\theta}(u)du 
$$
now the R.H.S.~is equal to $p_{\theta^{\star}}^{\epsilon}(y)$ if and only if $\theta=\theta^{\star}$ by (A\ref{hyp:2}-3); this proves 3. For 4.~consider
\begin{eqnarray*}
\log(p_{\theta}^{\epsilon}(y)) & = & \log\Bigg(\frac{1}{\int_{\{u\in\mathsf{Y}:\mathsf{d}(u,y)<\epsilon\}}du}\int_{\{u\in\mathsf{Y}:\mathsf{d}(u,y)<\epsilon\}}p_{\theta}(u)du\bigg) \\
& \leq & \frac{1}{\int_{\{u\in\mathsf{Y}:\mathsf{d}(u,y)<\epsilon\}}du}\int_{\{u\in\mathsf{Y}:\mathsf{d}(u,y)<\epsilon\}}\log(p_{\theta}(u))du \\
& \leq & C
\end{eqnarray*}
where we have applied Jensen's inequality to go to the second line and then used  (A\ref{hyp:2}-4) to go to the final line. Thus 4.~holds as $p_{\theta^{\star}}^{\epsilon}(y)$ is a probability density.
Hence we have proved that $\widetilde{\theta}_n^{\epsilon}\stackrel{a.s.}{\rightarrow}\theta^{\star}$. To prove that $\widehat{\theta}_n^{\epsilon}\stackrel{a.s.}{\rightarrow}\theta^{\star,\epsilon}$ 
one can apply \cite[Theorem 2.2]{white} for misspecified MLEs; assumptions (A1), (A2) and (A3a) (of that paper) follow by the above arguments (and our assumptions) and (A3b) follows by (A\ref{hyp:2}-5). This completes the proof of the proposition.
\end{proof}


\begin{thebibliography}{99}

\bibitem{amerin}
{\sc Amrein}, M. \& {\sc K\"unsch}, H.~(2011). A variant of importance splitting for rare event estimation: Fixed number of successes. \emph{TOMACS}, {\bf 21}, article 13.

\bibitem{andrieu2}\textsc{Andrieu}, C. \& \textsc{Roberts G.O.}(2009) The pseudo-marginal approach for efficient Monte Carlo computations,
\emph{Ann. Statist.,} \textbf{37}, 697--725. 

\bibitem{andrieu} 
\textsc{Andrieu}, C., \textsc{Doucet}, A. \& \textsc{Holenstein},
R.~(2010). Particle Markov chain Monte Carlo methods (with discussion).
\textit{J. R. Statist. Soc. Ser. B}, \textbf{72}, 269--342.

\bibitem{andrieu1}
\textsc{Andrieu}, C., \textsc{Doucet}, A. \& {\sc Lee}, A.~(2012).  Discussion of: Constructing summary statistics for approximate Bayesian computation: semi-automatic approximate Bayesian computation.
 \emph{J. Roy. Statist. Soc. Ser. B}, {\bf 74}, 451--452.

\bibitem{chopin1}
{\sc Barthelm\'e}, S. \& {\sc Chopin}, N. (2014). Expectation-Propagation for Summary-Less, Likelihood-Free Inference, \emph{J. Amer. Statist. Ass.}, (to appear).

\bibitem{bart}
{\sc Barthelm\'e}, S., {\sc Chopin}, N., {\sc Jasra}, A. \& {\sc Singh}, S. S.~(2012). Discussion of: Constructing summary statistics for approximate Bayesian computation: semi-automatic approximate Bayesian computation.
 \emph{J. Roy. Statist. Soc. Ser. B}, {\bf 74} 453--454

\bibitem{beaumont} 
\textsc{Beaumont}, M., \textsc{Zhang}, W. \& \textsc{Balding},
D.~(2002). Approximate Bayesian computation in population genetics.
\textit{Genetics}, \textbf{162}, 2025--2035.

\bibitem{biau}
{\sc Biau}, G., {\sc C\'erou}, F. \& {\sc Guyader}, A.~(2014). New insights into approximate Bayesian computation. \emph{Annales de l'IHP} (to appear).

\bibitem{blum}
{\sc Blum}, M.~(2010). Approximate Bayesian computation: a nonparametric perspective. \emph{J. Amer. Statist. Ass.}, {\bf 105}, 1178--1187.

\bibitem{bollerslev}
{\sc Bollerslev}, T.~(1986). Generalized Autoregressive Conditional Heteroskedasticity. \emph{J. Econom.}, {\bf 31}, 307--327.

\bibitem{briers}
{\sc Briers}, M., {\sc Doucet}, A., \& {\sc Maskell}, S.~(2010). Smoothing algorithms for state-space models.
\emph{Ann. Instit. Statist. Math}, {\bf 62}, 61--89.

\bibitem{boys}
{\sc Boys}, R. J., {\sc Wilkinson}, D. J. \& {\sc Kirkwood}, T. B. L.~(2008). Bayesian inference for a discretely observed stochastic network.
\emph{Statist. Comp.}, {\bf 18}, 125--135.

\bibitem{calvet}
{\sc Calvet}, C. \& {\sc Czellar}, V.~(2012).
Accurate Methods for Approximate Bayesian Computation Filtering. Technical Report, HEC Paris.

\bibitem{camros2009}
{\sc Campillo}, F. \& {\sc Rossi}, V. (2009). Convolution Particle Filter for Parameter Estimation
in General State-Space Models. \emph{IEEE Trans. Aero. Elec. Sys.}, {\bf 45}, 1073--1072.

\bibitem{cappe}
{\sc Capp\'e}, O., {\sc Ryden}, T, \& {\sc Moulines}, \'E.~(2005). \emph{Inference
in Hidden Markov Models}. Springer: New York.

\bibitem{cerou1}
{\sc C\'erou}, F., {\sc Del Moral}, P. \& {\sc Guyader}, A.~(2011).
A non-asymptotic variance theorem for un-normalized Feynman-Kac particle models. \emph{Ann. Inst. Henri Poincare}, {\bf 47}, 629--649. 

\bibitem{chambers}
{\sc Chambers}, J. M., {\sc Mallows}, C. L., \& {\sc Stuck}, B. W. (1976). Method for simulating stable random variables. 
\emph{J. Amer. Statist. Ass.}, {\bf 71}, 340--344.

\bibitem{chopin}
{\sc Chopin}, N. (2002). A
sequential particle filter for static models. \emph{Biometrika}, \textbf{89}, 539--552.

\bibitem{chopin_smc}
{\sc Chopin}, N., {\sc Jacob}, P. E. \& {\sc Papaspiliopoulos}, O.~(2013).
SM$\textrm{C}^2$: A sequential Monte Carlo algorithm with particle Markov chain Monte Carlo updates. 
\emph{J.~R.~Statist.} \emph{Soc. B}, {\bf 75}, 397-426.


\bibitem{cox} \textsc{Cox}, D. R.~(1981). Statistical analysis of
time-series: some recent developments. \emph{Scand. J. Statist.}~\textbf{8},
93--115.

\bibitem{dean1}
{\sc Dean}, T. \& {\sc Singh}, S.S. (2011) Asymptotic behaviour of approximate Bayesian estimators Technical Report, University of Cambridge.

\bibitem{dean}
{\sc Dean}, T., {\sc Singh}, S. S., {\sc Jasra}, A.,\& {\sc Peters}, G.~(2014).
Parameter inference for hidden Markov models with intractable likelihoods. {\it Scand. J. Statist.} (to appear).

\bibitem{delmoral1}
{\sc Del Moral}, P.~(2004). \textit{Feynman-Kac Formulae: Genealogical and
Interacting Particle Systems with Applications}. Springer: New York.

\bibitem{delmoral2}
{\sc Del Moral}, P.~(2013). \textit{Mean Field Simulation for Monte Carlo Integration.} Chapman \& Hall: London.


\bibitem{delmoraljacod}
{\sc Del Moral}, P., \& {\sc Jacod}, J.~(2002). The Monte-Carlo Method for filtering with discrete time observations. Central Limit Theorems.  
\emph{The Fields Institute Communications, Numerical Methods and Stochastics},  Ed. T.J. Lyons, T.S. Salisbury, Amer. Math. Soc.


\bibitem{delm:06}
{\sc Del Moral,} P., {\sc Doucet}, A. \& {\sc Jasra}, A.~(2006).
Sequential Monte Carlo samplers.
\emph{J.~R.~Statist.} \emph{Soc. B}, {\bf 68}, 411--436.


\bibitem{delmoral}
{\sc Del Moral}, P., {\sc Doucet}, A. \& {\sc Jasra}, A.~(2012).
An adaptive sequential Monte Carlo method for approximate Bayesian
computation. {\it Statist. Comp.}, {\bf 22}, 1009--1020.


\bibitem{dds1}
{\sc Del Moral}, P., {\sc Doucet}, A. \& {\sc Singh}, S. S.~(2010). 
A backward interpretation of Feynman-Kac formulae. \emph{M2AN}, {\bf 44},
947--975.


\bibitem{douc}
\textsc{Douc}, R., \textsc{Doukhan}, P. \& \textsc{Moulines},
E.~(2013). Ergodicity of observation-driven time series models and
consistency of the maximum likelihood estimator. \emph{Stoch. Proc. Appl.},
{\bf 123}, 2620--2647.

\bibitem{doucet}
{\sc Doucet}, A. \& {\sc Johansen}, A. M.~(2011). 
Particle filtering and smoothing: Fifteen years later. In \emph{Oxford Handbook of Nonlinear Filtering}, D. Crisan and B. Rozovsky (eds.).OUP: Oxford.

\bibitem{pitt} \textsc{Doucet}, A., \textsc{Pitt}, M., \& \textsc{Kohn},
R.~(2012). Efficient implementation of Markov chain Monte Carlo when
using an unbiased likelihood estimator. arXiv:1210.1871 {[}stat.ME{]}.

\bibitem{drov}
{\sc Drovandi}, C., {\sc Pettit}, A., \& {\sc McCutchan}, R.~(2013). Exact and approximate Bayesian inference for low count time series models with intractable likelihoods. Technical Report,
University of Queensland.


\bibitem{ehrlich1}
{\sc Ehrlich}, E.~(2013). \emph{On Inference for ABC Approximations
of Time Series Models}. PhD thesis, Imperial College London.

\bibitem{ehrlich}
{\sc Ehrlich}, E., {\sc Jasra}, A., \& {\sc Kantas}, N.~(2014). 
Gradient free parameter estimation for hidden
Markov models with intractable likelihoods. \emph{Method. Comp. Appl. Probab.} (to appear).

\bibitem{evensen}
{\sc Evensen}, G.~(1994). Sequential data assimilationwith a non-linear quasi geostrophic model using Monte Carlo methods to forecast error statistics. 
\emph{J. Geophys. Res.} {\bf 99}, 10143--10162.

\bibitem{fan}
{\sc Fan}, J. \& {\sc Yao}, Q.~(2005). \emph{Nonlinear Time Series: Nonparametric and Parametric Methods}.
Springer: New York. 

\bibitem{fearnhead}
{\sc Fearnhead}, P. \& {\sc Prangle}, D.~(2012).  Constructing summary statistics for approximate Bayesian computation: semi-automatic approximate Bayesian computation.
 \emph{J. Roy. Statist. Soc. Ser. B}, {\bf 74}, 419--474.

\bibitem{ferguson}
{\sc Ferguson}, T.~S.~(1996). \emph{A Course in Large Sample Theory}. Chapman \& Hall: London.

\bibitem{gauchi}
{\sc Gauchi}, J.-P. \& {\sc Vila}, J.-P.~(2013). Nonparametric filtering approaches for identification and inference in nonlinear
dynamic systems. \emph{Statist. Comp.} {\bf 23}, 523--533.

\bibitem{golightly}
{\sc Golightly}, A. \& {\sc Wilkinson} D. J.~(2011). Bayesian parameter inference
for stochastic biochemical network models using particle Markov chain Monte Carlo. \emph{Inter. Focus}, {\bf 6}, 807--820.

\bibitem{gour}
{\sc Gouri\'eroux}, C. \& {\sc Ronchetti}, E.~(1993). Indirect inference. \emph{J. Appl. Econmetr.}, {\bf 8}, s85--s118.

\bibitem{granger}
{\sc Granger}, C. W. J. \& {\sc Anderson}, A.~(1978) \emph{An Introduction to Bilinear Time Series Models}.
Vandenhoeck and Ruprecht: Gottingen.

\bibitem{haynes}
{\sc Haynes}, M.~(1998). \emph{Flexible Distributions and Statistical Models in Ranking and Selection Procedures, with Applications}. PhD thesis, Queensland University of Technology. 

\bibitem{jacob}
{\sc Jacob}, P. E. \& {\sc Thiery}, A. T.~(2013). On non-negative unbiased estimators. arXiv preprint.

\bibitem{jasra3}
{\sc Jasra}, A.~(2012).  Discussion of: Constructing summary statistics for approximate Bayesian computation: semi-automatic approximate Bayesian computation.
 \emph{J. Roy. Statist. Soc. Ser. B}, {\bf 74}, 461.

\bibitem{jasra1}
{\sc Jasra}, A., {\sc Kantas}, N., \& {\sc Ehrlich}, E.~(2014).
Approximate inference for observation driven time series
models with intractable likelihoods. {\it TOMACS} (to appear).


\bibitem{jasra2}
{\sc Jasra}, A., {\sc Lee}, A., {\sc Yau}, C \& {\sc Zhang}, X.~(2013). The alive particle filter. arXiv preprint.

\bibitem{jasra}
{\sc Jasra}, A., {\sc Singh}, S. S.,  {\sc Martin}, J. S., \& {\sc McCoy}, E.~(2012).
Filtering via Approximate Bayesian computation. {\it Statist.
Comp}, {\bf 22}, 1223--1237.

\bibitem{kanagawa}
{\sc Kanagawa}, M., {\sc Nishiyama}, Y., {\sc Gretton}, A., \& {\sc Fukumizu}, K.~(2013). Kernel Monte Carlo filter. arXiv preprint.

\bibitem{kantas}
{\sc Kantas}, N., {\sc Beskos}, A. \& {\sc Jasra}, A.~(2013). Sequential Monte Carlo methods for high-dimensional inverse problems: A case study for the Navier-Stokes equations. arXiv preprint.

\bibitem{kim}
{\sc Kim}, S., {\sc Shephard}, N. \& {\sc Chib}, S.~(1998). Stochastic volatlity: Likelihood inference and comparison with ARCH models. \emph{Rev. Econ. Stud.}, {\bf 65}, 361--393.

\bibitem{legland1}
{\sc Le Gland}, F. \& {\sc Oudjane}, N.~(2006). A sequential particle algorithm that keeps the particle system alive. In \emph{Stochastic Hybrid Systems : Theory and Safety Critical Applications}, (H. Blom \& J. Lygeros, Eds), Lecture Notes in Control and Information Sciences {\bf 337}, 351--389, Springer: Berlin

\bibitem{lee} \textsc{Lee}, A.~(2012).  On the choice of MCMC kernels
for approximate Bayesian computation with SMC samplers. \emph{In Proc.
Winter Sim. Conf.}.


\bibitem{lee1} \textsc{Lee}, A. \& \textsc{Latuszynski}, K.~(2012).
Variance bounding and geometric ergodicity of Markov chain Monte Carlo
for approximate Bayesian computation.  arXiv:1210.6703 {[}stat.ME{]}.

\bibitem{liu}
{\sc Liu} J.~(1990). A note on causality and invertibility of general bilinear time series models. \emph{Adv. Appl. Probab.}, {\bf 22}, 247--250.

\bibitem{majoram}
{\sc Majoram}, P., {\sc Molitor}, J., {\sc Plagnol}, V. \& {\sc Tavar\'e}, S.~(2003). Markov chain Monte Carlo without likelihoods. \emph{PNAS}, {\bf 100}, 15324--15328.

\bibitem{marin} \textsc{Marin}, J.-M., \textsc{Pudlo}, P., \textsc{Robert},
C.P. \& \textsc{Ryder}, R. (2012). Approximate Bayesian computational
methods. \emph{Statist. Comp.}, \textbf{22}, 1167--1180.


\bibitem{martin}
{\sc Martin}, J. S., {\sc Jasra}, A., {\sc Singh}, S. S., {\sc Whiteley}, N., {\sc Del Moral}, P. \& {\sc McCoy}, E.~(2014). Approximate Bayesian computation
for smoothing, \emph{Stoch. Anal. Appl.} (to appear).

\bibitem{mckcoodea2009}
{\sc McKinley}, J., {\sc Cook}, A. \& {\sc Deardon}, R. (2009). Inference for Epidemic Models With-
out Likelihooods. \emph{Intl. J. Biostat.}, {\bf 5}.

\bibitem{minka}
{\sc Minka}, T.~(2001). Expectation-Propagation for approximate Bayesian inference. \emph{Proc. UAI}, {\bf 17}, 362--369.

\bibitem{murray}
{\sc Murray}, L. M., {\sc Jones}, E. M. \& {\sc Parslow}, J. (2013). On disturbance state-space models and the particle marginal Metropolis-Hastings sampler. \emph{SIAM/ASA J. Unc. Quant.}, {\bf 1}, 494--521.

\bibitem{nott}
{\sc Nott}, D., {\sc Marshall}, L. \& {\sc Ngoc}, T. M.~(2012).
The ensemble Kalman filter is an ABC algorithm. \emph{Statist. Comp.}, {\bf 22}, 1273--1276.

\bibitem{oliver}
{\sc Oliver}, D. S., {\sc Cunha}, L. B., \& {\sc Reynolds}, A.C.~(1997). Markov chain Monte Carlo methods for conditioning a permeability field to pressure data. \emph{Math. Geol.}, {\bf 29}, 61--91. 

\bibitem{persing}
{\sc Persing}, A. \& {\sc Jasra} A.~(2013). Twisting the alive particle filter. arXiv preprint.

\bibitem{peters}
{\sc Peters}, G., {\sc Sisson}, S. \& {\sc Fan} Y.~(2012). Likelihood-free Bayesian inference for alpha-stable models. \emph{Comp. Stat. Data Anal.}, {\bf 56}, 3743--3756.

\bibitem{poya}
{\sc Poyiadjis}, G., {\sc Doucet}, A., \& {\sc Singh}, S. S.~(2011).
Particle approximations of the score and observed information matrix in state space models with application to parameter estimation. \emph{Biometrika}, {\bf 98},  65–-80.

\bibitem{pritchard}
{\sc Pritchard}, J., {\sc Sielstad}, M., {\sc Perez-Lezaun}, A., \& {\sc Feldman}, M.~(1999). Population growth of human Y chromosomes: a study of Y chromosome microsatellites. \emph{Mol. Biol. Evol.}, {\bf 16}(12), 1791--1798.

\bibitem{robert}
{\sc Robert}, C. P.~(2012). Discussion of: Constructing summary statistics for approximate Bayesian computation: semi-automatic approximate Bayesian computation.
 \emph{J. Roy. Statist. Soc. Ser. B}, {\bf 74}, 447--448.

\bibitem{robert1}
{\sc Robert}, C. P. \& {\sc Casella}, G.~(2004). \emph{Monte Carlo Statistical Methods}. 2nd edition, Springer: New York.

\bibitem{robert2}
{\sc Robert}, C.P., {\sc Cornuet}, J.-M., {\sc Marin}, J.-M. \&  {\sc Pillai}, N.S.~(2011). Lack of confidence in approximate Bayesian computational (ABC) model choice. \emph{PNAS}, {\bf 108}(37), 15112--15117. 

\bibitem{rubin}
{\sc Rubin}, D.~(1984). Bayesianly jutifiable and relevant frequency calculations for the applied statistician. \emph{Ann. Statist.}, {\bf 12}, 1151--1172.

\bibitem{adam} 
\textsc{Rubio}, F. \& \textsc{Johansen}, A.~(2013).
A simple approach to maximum intractable likelihood estimation.
\emph{Elec. J. Statist.}, {\bf 7} 1632-1654.


\bibitem{sherlock}
{\sc Sherlock}, C., {\sc Thiery}, A., {\sc Roberts}, G. O. \& {\sc Rosenthal}, J.~(2013)
On the efficiency of pseudo-marginal random walk Metropolis algorithms. arXiv preprint.

\bibitem{wheedon}
{\sc Wheedon}, R. H. \& {\sc Zygmund}, A.~(1977). \emph{Measure and Integral: An Introduction to Real Analysis}. Marcel Dekker: New York.

\bibitem{white}
{\sc White}, H.~(1982). Maximum likelihood estimation of misspecified models. \emph{Econometrica}, {\bf 50}, 1--25.

\bibitem{theo}
{\sc White}, S. R.~, {\sc Kypraios}, T. \& {\sc Preston}, S. T.~ (2014). Piecewise Approximate Bayesian Computation: fast inference
for discretely observed Markov models using a factorised
posterior distribution. \emph{Statist. Comp.} (to appear).

\bibitem{wil2008}
{\sc Wilkinson}, R.D. (2013). Approximate Bayesian computation (ABC) gives exact results under the assumption of model error. \emph{Statist. Appl. Genetics  Mole. Biol.}, {\bf 12}, 129--141. Also available as {\tt arXiv:0811.3355}, 2008.

\bibitem{yildrlim}
{\sc Yildirim}, S., {\sc Singh}, S.S., {\sc Dean}, T.~\& {\sc Jasra}, A.~(2013).
Parameter estimation in hidden Markov models with intractable likelihoods using sequential Monte Carlo. arXiv preprint.

\bibitem{zhang}
{\sc Zhang}, X.~(2013). \emph{Some Contibutions to Approximate
Inference in Bayesian Statistics}. PhD thesis, Imperial College London.

\bibitem{zhou}
{\sc Zhou}, Y.,  {\sc Johansen}, A. M.  \&  {\sc Aston}, J. A. D.~(2013). Towards Automatic Model Comparison: An Adaptive Sequential Monte Carlo Approach. arXiv preprint.

\end{thebibliography}
\end{document}